\documentclass[journal,12pt,draftclsnofoot,onecolumn]{IEEEtran}
\usepackage[usenames,dvipsnames]{color}
\usepackage{amsmath}
\usepackage{amssymb}
\usepackage{amsfonts}
\usepackage{mathrsfs}
\usepackage{graphicx}
\usepackage{epstopdf}
\usepackage{subfig}
\usepackage[noadjust]{cite} 
\usepackage{url}
\usepackage{amsthm}
\theoremstyle{plain}
\newtheorem{theorem}{Theorem}
\theoremstyle{definition}

\theoremstyle{definition}
\newtheorem{Remarks}{Remark}
\theoremstyle{plain}
\newtheorem{lemma}{Lemma}
\theoremstyle{plain}
\newtheorem{proposition}{Proposition}

\begin{document}

\title{Nonparametric Decentralized Sequential Detection via Universal Source Coding}
\author{Jithin~K.~Sreedharan and Vinod~Sharma,~\IEEEmembership{Senior~Member,~IEEE}
\thanks{Jithin K.\ Sreedharan is with INRIA Sophia Antipolis, France (email: jithin.sreedharan@inria.fr)}
\thanks{Vinod Sharma is with Department of Electrical Communication Engineering, Indian Institute of Science, Bangalore, India (email: vinod@ece.iisc.ernet.in)}
\thanks{Preliminary versions of this paper have been presented in Allerton 2011, Globecom 2012 and ITA 2013.}}
\maketitle

\begin{abstract}
We consider nonparametric or universal sequential hypothesis testing problem when the distribution under the null hypothesis is fully known but the alternate hypothesis corresponds to some other unknown distribution. These algorithms are primarily motivated from spectrum sensing in Cognitive Radios and intruder detection in wireless sensor networks. We use easily implementable universal lossless source codes to propose simple algorithms for such a setup. The algorithms are first proposed for discrete alphabet. Their performance and asymptotic properties are studied theoretically. Later these are extended to continuous alphabets. Their performance with two well known universal source codes, Lempel-Ziv code and Krichevsky-Trofimov estimator with Arithmetic Encoder are compared. These algorithms are also compared with the tests using various other nonparametric estimators. Finally a decentralized version utilizing spatial diversity is also proposed. Its performance is analysed and asymptotic properties are proved.
\end{abstract}


\section{Introduction}
Distributed detection (\cite{Viswanathan_IEEE-1}) in which different local nodes interact each other to make a final decision, has been quite popular recently due to its relevance to distributed radar, sensor networks (\cite{Chamberland_IEEESPM07}), distributed databases and cooperative spectrum sensing in Cognitive radios (\cite{Akyildiz_PC2011,Jithin_WCNC2011,Jithin_JSAC-submitted_arxiv}). It decreases probability of errors and detection delay by making use of spatial diversity and mitigates the effects of multipath fading, shadowing and hidden node problem experienced in single node detection.

Distributed detection can use either decentralized or centralized algorithms. In the centralized framework, the information received by the local nodes or sensors are transmitted directly to the fusion center (FC) to decide upon the hypothesis. In decentralized detection each local node sends a summarized or quantized information to the fusion center (\cite{Viswanathan_IEEE-1, Chamberland_IEEESPM07}). The fusion center ultimately decides upon which hypothesis is true. Thus decentralized detection although suboptimal, is more bandwidth and energy efficient. A drawback of a decentralized scheme is that the fusion center makes the decision based on less information. Hence the main challenge of decentralized detection algorithms is to provide a reliable decision with this information. Performance depends on the local node and fusion node detection policies and the type of feedback from the fusion node to the local nodes. The main resource constraints for decentralized detection schemes include number of nodes, finite alphabet constraint on output of each local node, limited spectral bandwidth, total cost of the system and stringent power requirements.

Two of the important formulations of distributed detection problem are based on the number of samples required for making a decision: fixed sample size and sequential detection (\cite{Lai_SS2001,Siegmund_SATC_Book}). In fixed sample size detection, the likelihood ratio test on the received data minimises the probability of error at the fusion center for a binary hypothesis testing problem. Hence the real problem in this case is to decide the type of information each local node should send to the fusion center. Interestingly likelihood tests at the local nodes are optimal whenever the observations are conditionally independent, given the hypothesis (\cite{Chamberland_IEEESPM07}).

In the sequential case, the observations are sampled sequentially at the local nodes until a stopping rule is satisfied. The decision and stopping rules are designed to reduce the number of samples for decision making with reliability constraints. More precisely, sequential detectors can detect change in the underlying hypothesis or test the hypothesis (\cite{Lai_SS2001,Siegmund_SATC_Book}). In this paper we focus on decentralized sequential hypothesis testing. It is well known that in case of a single node, Sequential Probability Ratio Test (SPRT) out-performs other sequential  or fixed sample size detectors (\cite{Siegmund_SATC_Book}). In case of decentralized setup, optimization needs to be performed jointly over the local nodes and the fusion center policies as well as over time. Unfortunately, this problem is intractable for most scenarios (\cite{Mei_TIT2008,Veeravalli1999}). Specifically there is no optimal solution available when there is no feedback from the fusion center and there is limited local memory, which are more relevant in practical situations. In parametric case, \cite{Fellouris_TIT2011} and \cite{Mei_TIT2008} proposed asymptotically optimal (order 1 (Bayes) and order 2 respectively) decentralized sequential hypothesis tests for such systems with full local memory. But these models do not consider noise at the fusion center and assume a perfect communication channel between the local nodes and the fusion center. Noisy channels between local nodes and fusion center are considered in \cite{Yilmaz_Allerton2012_journal}. Also recently \cite{Jithin_JSAC-submitted_arxiv} studied the setup when the communication channel between the local nodes and the FC is a noisy MAC (Multiple Access Channel) and nearly asymptotically optimal algorithms are derived.

Sequential methods in case of uncertainties are surveyed in \cite{Lai_SS2001} for a parametric family of distributions. For nonparametric sequential methods, \cite{Nitis_SMTA_Book} provides separate algorithms for different setups like changes in mean, changes in variance etc. In this paper we propose unified simple sequential hypothesis testing algorithms using universal source coding (\cite{Cover_EIT_book}) where the unknown alternate distribution can belong to a nonparametric family and study its properties.

An optimal fixed sample size universal test for finite alphabets is derived in \cite{Hoeffding_AMS1965}. Error exponents for these tests are studied in \cite{Levitan_TIT2002}. In \cite{Unnikrishnan_TIT2011} mismatched divergence is used to study this problem. Statistical inference with universal source codes, started in \cite{Rissanen_TIT1984} where classification of finite alphabet sources is studied in the fixed sample size setup. \cite{Tony_ISIT2008} considers the universal hypothesis testing problem in the sequential framework using universal source coding. It derives asymptotically optimal one sided sequential hypothesis tests and sequential change detection algorithms for countable alphabet. In one sided tests one assumes null hypothesis as the default hypothesis and has to wait a long time to confirm whether it is the true hypothesis (it is the true hypothesis only when the test never stops). In many practical applications where it is important to make a quick decision (e.g., in Cognitive Radios) this setup is not suitable.

In this paper, we consider universal source coding framework for  binary hypothesis sequential testing with continuous alphabets. Section \ref{sec:model} describes the model. Section \ref{sec:ch6:uni_seq_fin_alp} provides our algorithm for a single node with finite alphabet. We prove almost sure finiteness of the stopping time. Asymptotic properties of probability of error and moment convergence of expected stopping times are also studied. Section \ref{sec:ch6:uni_seq_con_alp} extends the test to continuous alphabet. Algorithms based on two well known easily implementable universal codes, Lempel-Ziv tree-structured (LZ78) (\cite{Ziv_TIT1978}) codes and Krichevsky-Tofimov estimator with Arithmetic Encoder (KT-AE) (\cite{Csiszar_ITS_book}) are studied in Section \ref{sec:ch6:lzslrt-ktslrt}. Performance of these tests are compared in Section \ref{sec:ch6:uni_seq_per_com}. In Section \ref{sec:ch6:uni_seq_dec_det} we extend our algorithm to the decentralized scenario. In our distributed algorithm each local node sends a local decision to the FC at asynchronous times leading to considerable saving in communication cost. Previous works in decentralized framework
(see, e.g., \cite{Akyildiz_PC2011} for Cognitive Radios) does not consider the universal setup, to the best of our knowledge. An approximate analysis of the decentralized algorithm is also presented here. Section \ref{sec:DualSPRT_asy_opt} explores the asymptotic properties of the decentralized test. Section \ref{sec:ch6:uni_seq_con} concludes the chapter. 

 \section{Model for single node}
\label{sec:model}
We consider the following hypothesis testing problem: Given i.i.d.\ observations $X_1,X_2,\ldots,$ we want to know whether these observations came from the distribution ${P_0}$ (hypothesis $H_0$) or from some other distribution ${P_1}$ (hypothesis $H_1$). We will assume that $P_0$ is known but $P_1$ is unknown. 

Our problem is motivated from the Cognitive Radio spectrum sensing (\cite{Akyildiz_PC2011}) and wireless sensor network intruder detection scenario (\cite{WSN_Swami2007}). Then usually $P_0$ is fully known (e.g., when licensed user is not transmitting in Cognitive Radios). However, under $H_1$, $P_1$ will usually not be completely known to the local node (e.g., with unknown licensed user, transmission parameters and channel gains).   

We first discuss the problem for a single node and then generalize to decentralized setting. Initially we study the case when $X_k$ take values in a finite alphabet. However, we will be mainly concerned with continuous alphabet observations because the receiver almost always has Gaussian noise. This will be taken up in Section \ref{sec:ch6:uni_seq_con_alp}

For convenience we summarize the important notation introduced for the single node algorithms in Table \ref{tab:ch3:list_symbols_sing} and for the decentralized algorithms in Sections \ref{sec:ch6:uni_seq_dec_det} and \ref{sec:DualSPRT_asy_opt} in Table \ref{tab:ch3:list_symbols_dec}. 

{\renewcommand{\arraystretch}{0.95}
\begin{table}[!tbh]
\begin{center}
\begin{tabular}[t]{|c||l|}\hline
  Notation & Meaning\\ \hline \hline
  $X_{k}$ & Observation at time $k$  \\ \hline
  $X_{k}^{\Delta}$ & Uniformly quantized observation of $X_k$ at time $k$ with quantization step $\Delta$ \\ \hline
  $P_i$, $f_{i}$, $f_i^{\Delta}$ & Probability distribution, PDF, PMF after quantization under $H_i$\\ \hline
  $\overline{H}_i$ & Entropy rate under $H_i$\\ \hline
  $W_n$ & Test statistic of SPRT at time $n$ \\ \hline
  $\widehat{W}_{n}$ & Test statistic of finite alphabet algorithm at time $n$ \\ \hline
  $\widetilde{W}_n$ & Test statistic of continuous alphabet algorithm at time $n$ \\ \hline
  $\widetilde{W}_n^{LZ}$, $\widetilde{W}_n^{KT}$ & $\widetilde{W}_n$ for LZSLRT and KTSLRT, respectively\\ \hline
  $\log \beta, -\log \alpha$ & Thresholds, $0<\alpha,\beta<1$\\ \hline
  $N$ & First time test statistic crosses $(\log \beta, -\log \alpha)$ \\ \hline
  $N^1, N^0$ & First time test statistic crosses $-\log \alpha$, crosses $\log \beta$, respectively\\ \hline
  $L_n(X_1^n)$ & Length of the codeword of the universal source code for the data $X_1, \ldots,X_n$\\ \hline
  $\lambda$ & Design parameter, related to minimum SNR under consideration \\ \hline
  $\mathcal{C}$ & Class of $P_1$, $\{P_1:D(P_1||P_0) \geq \lambda\}$ \\ \hline
  $N_i^*(\epsilon)$ & $\sup \{ {n\geq 1}: |-L_n(X_1^n)-\log P_i(X_1^n)|> n \epsilon \}$ \\ \hline
  $\delta$, $\rho^2$ & $E_1[\log P_1(X_1)/P_0(X_1)]-\lambda/2$, $Var_1[\log P_1(X_1)/P_0(X_1)]$ \\ \hline\
  $|A|$ & Alphabet size of the quantized alphabet \\ \hline
\end{tabular}\caption{List of important notations introduced in the single node algorithms.}
\label{tab:ch3:list_symbols_sing}
\end{center}
\vspace{-1.2 cm}
\end{table}}

{\renewcommand{\arraystretch}{0.95}
\begin{table}{!htbp}
\begin{center}
\begin{tabular}[t]{|c||l|}\hline
  Notation & Meaning\\ \hline \hline
  $L$ & Number of local nodes \\ \hline
  $X_{k,l}$ & Observation at node $l$ at time $k$  \\ \hline
  $Y_{k,l}$ & Transmitted value from node $l$ to FC at time $k$.\\ \hline
  $Y_k$ & FC observation at time $k$ \\ \hline
  $Z_k$ & FC MAC noise at time $k$ \\ \hline
  $f_{i,l}$, $g_{\mu}$ & PDF of $X_{1,l}$ under $H_i$, PDF of $Z_k+\mu$\\ \hline
  $\widehat{W}_{k,l}$ & Test statistic at local node $l$ at time $k$ \\ \hline
  $F_k$ & Test statistic at FC at time $k$ \\ \hline
  $\xi_k$ & LLR at FC \\ \hline
  $\xi_k^*$ & LLR when all nodes transmit wrong decisions \\ \hline
  $\theta_i$ & Worst case value of $E_i[\xi_k],=E_i[\xi_k^*]$ \\ \hline
  $F_n^*$, $\widehat{F}_n^*$& $\sum_{k=1}^{n}\xi_k^*$, $\sum_{k=1}^{n}|\xi_k^*|$ \\ \hline
  $\mathcal{A}^i$, $\Delta(\mathcal{A}^i)$ & \{all nodes transmit $b_i$ under $H_i$ \}, $E_i[\xi_k|\mathcal{A}^i]$ 
\\ \hline  
  $\gamma_{1,l}, \gamma_{0,l}$ & Thresholds at local node $l$\\ \hline
  $\beta_1, \beta_0$ & Thresholds at FC \\ \hline
  $\mu_1,\mu_0$ & Design parameters in FC LLR \\ \hline
  $b_1,b_0$ & Transmitting values to the FC from the local node \\ \hline
  $N_d$ & First time $F_k$ crosses $(-\beta_0, \beta_1)$ \\ \hline
  $N^1_d, N^0_d$ & First time $F_k$ crosses $\beta_1$, crosses $-\beta_0$ \\ \hline
  $N_l,N_l^1,N_l^0$ & Corresponding values of $N$, $N^1$, $N^0$ at local node $l$ \\ \hline
  $N_{0,l}^*(\epsilon)$ & $N_0^*(\epsilon)$ at local node $l$ \\ \hline
  $\delta_{l},\rho^2_{l}$ & Mean and variance of LLR at node $l$ under $H_1$ \\ \hline
  $\delta_{i,FC}^j$ & Mean of LLR at FC under $H_i$ when $j$ nodes transmit \\ \hline
  $t_j$ & Time epoch when $\delta_{i,FC}^{j-1}$ changes to $\delta_{i,FC}^j$ $(1)$\\ \hline
  $\bar{F_j}$ & $E[F_{t_j-1}]$\\ \hline
  $D_{tot}^0$, $D_{tot}^1$ & $L\lambda/2$, $\sum_{l=1}^L (D(f_{1,l}||f_{0,l})-\lambda/2)$\\ \hline
  $r_l$, $\rho_l$ & $\lambda/2D_{tot}^0$, $(D(f_{1,l}||f_{0,l})-\lambda/2)/D_{tot}^1$\\ \hline
  $\tau_l^*(a)$ & Last time $\widehat{W}_{n,l}$ will be above $a$\\ \hline
  $\tau^*(a)$ & $\displaystyle \max_{1\leq l \leq L} \tau^*_l(a)$ \\ \hline
  $\tau_l(a)$ & Last time a RW with drift $-\lambda/2-\epsilon$ at node $l$ will be above $a$ \\ \hline
  $R_i$ & $\displaystyle\min_{1\leq l \leq L} -\log \inf_{t \geq 0} E_i\Big[\exp \left({-t \log \frac{f_{1,l(X_{1,l})}}{f_{0,l}(X_{1,l})}}\right)\Big]$ \\ \hline
  $G_i$, $\widehat{G}_i$, $g_i$, $\widehat{g}_i$ & CDF of $|\xi_1^*|$, $\xi_1^*$, MGF of $|\xi_1^*|$, $\xi_1^*$\\ \hline
  $\Lambda_i(\alpha)$, $\widehat{\Lambda}_i(\alpha)$ & $\sup_{\lambda}(\alpha \lambda-\log g_i(\lambda))$, $\sup_{\lambda}(\alpha \lambda-\log \widehat{g}_i(\lambda))$ \\ \hline
  $\alpha_i^+$ & $\text{ess }\sup |\xi_1^*|$ \\ \hline
  & First time RW \\
  $\nu(a)$ & $\{\log \frac{g_{\mu_1(Z_k)}}{g_{-\mu_0(Z_k)}}+(\Delta(\mathcal{A}^0)-E_0[\log \frac{g_{\mu_1(Z_k)}}{g_{-\mu_0(Z_k)}}])$  \\
  & $ k \geq \tau(c)+1\}$ crosses $a$.\\ \hline
\end{tabular}\caption{List of important notations introduced in the decentralized algorithm.}
\label{tab:ch3:list_symbols_dec}
\end{center}
\end{table}}

\section{Finite Alphabet}
\label{sec:ch6:uni_seq_fin_alp}

We first consider finite alphabet for the distributions $P_0$ and $P_1$.

A sequential test is usually defined by a stopping time $N$ and a decision rule $\delta$. For SPRT (\cite{Siegmund_SATC_Book}), 
\begin{equation}
\label{eq:ch6:sprt_expn}
N\stackrel{\Delta}=\inf\{n:W_n \notin (\log \beta,-\log \alpha)\},\quad 0<\alpha ,\, \beta <1,
\end{equation}
\hspace{0.0 cm} where,
\begin{equation}
\label{eq:ch6:sprt_LLR}
 W_n=\sum_{k=1}^{n}\log \frac{P_1(X_{k})}{P_0(X_{k})}.
\end{equation}
At time $N$, the decision rule $\delta$ decides $H_1$ if $W_N \geq -\log \alpha$ and $H_0$ if $W_N \leq \log \beta$. 

SPRT requires full knowledge of $P_0$ and $P_1$. Now we propose our test when $P_1$ is unknown by replacing the log likelihood ratio process $W_n$ in (\ref{eq:ch6:sprt_LLR}) by
\begin{equation}
\label{eq:ch6:two_sided_finite_alphabet_LLR}
\widehat{W}_n=-L_n(X_1^n)-\log P_0(X_1^n)-n\frac{\lambda}{2},\quad \lambda >0,
\end{equation}
where $\lambda > 0$ is an appropriately chosen constant and $L_n(X_1^n)$ is the length of the codeword for data $X_1^n\stackrel{\Delta}=X_1,\ldots,X_n$ for a selected universal source code. We may recall that a universal source code does not need the distribution of $X_1,\ldots,X_n$

The following discussion motivates our test:
\begin{enumerate}
\item 
By Shannon-Macmillan Theorem (\cite{Cover_EIT_book}), $\displaystyle \lim_{n\to\infty} \allowbreak n^{-1}\log P(X_1^n) \allowbreak=-\overline{H}(X)$ for any stationary, ergodic source  a.s.\ where $\overline{H}(X)$ is the entropy rate. We consider universal lossless codes whose codelength function $L_n$ satisfies $\lim_{n\to\infty} n^{-1} L_n = \overline{H}(X)$ a.s., at least for i.i.d sources. The codes which satisfy this condition are called pointwise universal whereas the codes which satisfy this in terms of expectation are called universal. It is shown in \cite{WEI_UNPUB} that not all universal codes are pointwise universal. We consider algorithms like LZ78 (\cite{Ziv_TIT1978}) and KT-AE (\cite{Krichevsky_TIT1981}) which satisfy this convergence. Thus, for such universal codes,
\begin{equation}
\label{eq:ch6:pointwise_universality}
\frac{1}{n}(L_n(X_1^n)+\log P(X_1^n)) \to 0 \text{ w.p.1}.
\end{equation}
\item Under hypothesis $H_1$, $E_1[-\log P_0(X_1^n)]$ is approximately $nH_1(X)+nD(P_1||P_0)$ and for large $n$, $L(X_1^n)$ is approximately $nH_1(X)$ where $H_1(X)$ is the entropy under $H_1$ and $D(P_1||P_0)$ is the KL-divergence defined for two probability distributions $P$ and $Q$ on the same measurable space $(\Omega,\mathcal{F})$ as
\begin{equation}
\label{eq:ch3:KL-divergence}
D(P||Q)=\left\{
\begin{array}{ll}
\int \log \frac{dP}{dQ} \, dP &, \text{ if } P<<Q,\\
\infty &, \text{ otherwise ,}
\end{array}
\right.
\end{equation}
where $P<<Q$ denotes that $P$ is absolutely continuous w.r.t.\ $Q$. The above approximation gives the average drift in \eqref{eq:ch6:two_sided_finite_alphabet_LLR} under $H_1$ as $D(P_1||P_0)-\lambda/2$ and under $H_0$ as $-\lambda/2$. To get some performance guarantees (average drift under $H_1$ greater than $\lambda/2$), we limit $P_1$ to a class of distributions,
\begin{equation}
\label{eq:ch6:class_c_expn}
\mathcal{C}=\{P_1:D(P_1||P_0)\geq \lambda\}.
\end{equation}
where $\lambda$ is related to the minimum SNR under consideration. Divergence has been used in statistics in many different scenarios (\cite{Csiszar_ITS_book,Cover_EIT_book}).
\item When considering universal hypothesis testing in Neyman-Pearson framework (fixed sample size) the existing work considers the optimisation problem in terms of error exponents (\cite{Levitan_TIT2002}):
{\allowdisplaybreaks
\begin{eqnarray}
\sup_{\delta_{FSS}}\,\liminf_{n\to \infty} -\log P_{MD},\nonumber \\
\label{eq:ch6:FSS_error_exp_opt_crit}
\text{ such that}\quad \liminf_{n\to \infty} -\log P_{FA} \geq \hat{\alpha},
\end{eqnarray}}
\noindent where $P_{FA}$ is the false alarm probability, $P_{MD}$ is the miss-detection probability, $\delta_{FSS}$ is the fixed sample size decision rule and $\hat{\alpha}>1$. But in the sequential detection framework the aim is to
\[\min_{(N,\delta)} E_1[N],\quad \min_{(N,\delta)} E_0[N],\]
\[\textrm{such that}\quad P_{FA} \leq \alpha \textrm{ and } P_{MD} \leq \beta.\]
In case of the universal sequential detection framework, the objective can be to obtain a test satisfying $P_{FA} \leq \alpha \textrm{ and } P_{MD} \leq \beta$ with
{\allowdisplaybreaks
\begin{eqnarray}
\label{eq:ch6:asymp_optimality_1}
\lim_{\alpha+\beta \to 0} \frac{E_1[N]}{|\log \alpha|}=\lim_{\alpha+\beta \to 0} \frac{E_1^{S}[N]}{|\log \alpha|}=\frac{1}{D(P_1||P_0)},\\
\label{eq:ch6:asymp_optimality_2}
\lim_{\alpha+\beta \to 0} \frac{E_0[N]}{|\log \beta|}=\lim_{\alpha+\beta \to 0} \frac{E_0^{S}[N]}{|\log \beta|}=\frac{1}{D(P_0||P_1)},
\end{eqnarray}}
where $E_i^{S}(N)$ is the expected value of $N$ under $H_i$ for SPRT, $i=0,1$. We will study such results for our algorithm in Theorem \ref{thm:ch6:pfa_pmd_prtys} and Theorem \ref{thm:ch6:edd_prtys}.
\end{enumerate}

Thus our test is to use $\widehat{W}_n$ in (\ref{eq:ch6:sprt_expn}) when $P_0$ is known and $P_1$ can be any distribution in class $\mathcal{C}$ defined in (\ref{eq:ch6:class_c_expn}). Our test is useful for stationary and ergodic sources also. 

The following proposition proves the almost sure finiteness of the stopping time of the proposed test. This proposition holds if $\{X_k\}$ are stationary, ergodic and the universal code satisfies a weak pointwise universality. Let $\overline{H}_i$ be the entropy rate of $\{X_1,X_2,\ldots\}$ under $H_i,i=0,1$. Also let $N^1=\inf \{n:\widehat{W}_n \geq -\log \alpha \}$ and $N^0=\inf \{n:\widehat{W}_n \leq \log \beta \}$. Then $N=\min(N^0,N^1)$.
\begin{proposition}
\label{prop:ch6:as_st}
Let $L_n(X_1^n)/n \to \overline{H}_i$ in probability for $i=0,1$. Then
\begin{enumerate}
\item[(a)] $P_0 (N < \infty)=1$,
\item[(b)] $P_1 (N < \infty)=1$.
\end{enumerate}
\end{proposition}
\begin{proof} 
See Appendix \ref{proof:prop:ch6:as_st}.
\end{proof}
%

We introduce the following notation: for $\epsilon >0$ and for $i=0,1$,
{\allowdisplaybreaks
\begin{eqnarray}
N_i^*(\epsilon) &\stackrel{\Delta} =& \sup \{ {n\geq 1}: |-L_n(X_1^n)-\log P_i(X_1^n)|> n \epsilon \}.\label{eq:ch6:N_0(eps)_def}
\end{eqnarray}}
\indent Observe that $E_{P_1}({N_1^*(\epsilon)}^p)<\infty$ for all $\epsilon >0$ and all $p>0$ is implied by a stronger version of pointwise universality, 
\begin{equation}
\label{eq:stron_ver_po_uni}
\max_{x_1^n \in \mathcal{X}^n} \Big(L_n(x_1^n)+\log P_1(x_1^n)\Big)\sim o(n),
\end{equation}
$\mathcal{X}$ being the source alphabet. Similarly $E_{P_0}({N_0^*(\epsilon)}^p)<\infty$ also holds. This property is satisfied by KT-AE (\cite[Chapter 6]{Csiszar_ITS_book}) and LZ78 (\cite{Kieffer_DCC1999, Ziv_TIT1978}) encoders.

The following theorem gives a bound for $P_{FA}$ and an asymptotic result for $P_{MD}$.

\begin{theorem}
\label{thm:ch6:pfa_pmd_prtys}
\item[(1)] For prefix free universal codes, $P_{FA}\stackrel{\Delta}=P_0(\widehat{W}_N \geq -\log \alpha)\leq \alpha$.
\item[(2)] If the observations $X_1,X_2,\ldots,X_n$ are i.i.d.\ and the universal source code satisfies \eqref{eq:stron_ver_po_uni}, then \[P_{MD}\stackrel{\Delta}=P_1(\widehat{W}_N \leq \log \beta)= \mathcal{O}(\beta^{s}),\]
where $s$ is the solution of $\displaystyle E_1\Big[e^{-s\,\left(\log \frac{P_1(X_1)}{P_0(X_1)}-\frac{\lambda}{2}-\epsilon \right)}\Big]=1$ for $0<\epsilon < \lambda/2$ and $s>0$.
\end{theorem}

\begin{proof}
See Appendix \ref{proof:thm:ch6:pfa_pmd_prtys}.
\end{proof}

Under the above assumptions, we also have the following. We will use the notation, $\delta=E_1\left[\log \frac{P_1(X_1)}{P_0(X_1)}\right]-\frac{\lambda}{2}$, $\rho^2={Var}_{H_1}^2\left[\log \frac{P_1(X_1)}{P_0(X_1)}\right]$.
\begin{theorem}
\label{thm:ch6:edd_prtys}
\item[(a)] Under $H_0$, $\displaystyle \lim_{\alpha,\beta \to 0} \frac{N}{|\log \beta|} =\frac{2}{\lambda}$ a.s.\ If $E_{0}[{N_0^*(\epsilon)}^p]<\infty$ and $E_0[{(\log P_0(X_1))}^{p+1}]<\infty$ for all $\epsilon >0$ and for some $p \geq 1$, then also,
\[ \lim_{\alpha,\beta \to 0} \frac{E_0[(N)^q]}{{|\log \beta|}^q}=\lim_{\alpha,\beta \to 0} \frac{E_0[(N^0)^q]}{{|\log \beta|}^q}= {\left(\frac{2}{\lambda}\right)}^q, \]
for all $0< q \leq p$.
\item[(b)] Under $H_1$, $\displaystyle \lim_{\alpha,\beta \to 0}\frac{N}{|\log \alpha|} =\frac{1}{\delta}$ a.s.\ If $E_{1}[{N_1^*(\epsilon)}^p]<\infty$, $E_1[{(\log P_1(X_1))}^{p+1}]<\infty$ and $E_1[{(\log P_0(X_1))}^{p+1}]\allowbreak<\infty$ for all $\epsilon >0$ and for some $p \geq 1$, then also,
\[\lim_{\alpha,\beta \to 0} \frac{E_1[(N)^q]}{{|\log \alpha|}^q}=\lim_{\alpha,\beta \to 0} \frac{E_1[(N^1)^q]}{{|\log \alpha|}^q}= {\left(\frac{1}{\delta}\right)}^q, \]
for all $0< q \leq p$. 

If $\rho^2<\infty$ and $(L_n(X_1^n)+\log P_1(X_1^n))/\sqrt{n} \to 0$ a.s.\ then 
\begin{equation*}
\frac{N-{\delta}^{-1}|\log \alpha|}{\sqrt{\rho^2\delta^{-3}|\log \alpha|}} \to \mathcal{N}(0,1) \text{ in distribution.}
\end{equation*}
\end{theorem}

\begin{proof}
See Appendix \ref{proof:thm:ch6:edd_prtys}.
\end{proof}
\vspace{0.6 cm}

Comparing Theorem \ref{thm:ch6:edd_prtys} with \eqref{eq:ch6:asymp_optimality_1} and \eqref{eq:ch6:asymp_optimality_2} shows that our result has the optimal rate of convergence, although the limiting value for our test may be somewhat larger than that of SPRT.

From Remark \ref{rem:sup_perf_KTSLRT_arg} below, we will see that KT-AE satisfies the conditions for central limit theorem in Theorem \ref{thm:ch6:edd_prtys}(b) but not LZ78.

Table \ref{tab:ch6:comp_edd_analysis_simln} shows that the asymptotics for $E_1[N]$ and $E_0[N]$ match with simulations well at low probability of error. In the table $P_0\sim B(8,0.2)$ and $P_1\sim B(8,0.5)$, where $B(n,p)$ represents Binomial distribution with $n$ number of trials and $p$ success probability in each trial. Also $\lambda=1.2078$. We use the KT-AE, which is presented in Section \ref{subsec:ch6:KTSLRT_simulation}, as the universal source code.

\begin{table}[!tbh]
\begin{center}
\begin{tabular}[t]{|c|c|c|c|}\hline
$Hyp=i$ & {$P_{i}(H_j),j \neq i$} & {$E_{i}[N] \text{ Theory}$} & {$E_{i}[N]\text{ Simln.}$} \\ \hline
$0$  & $3e-4$ & $47.6$ & $52.2$ \\ \hline
$0$  & $5e-6$ & $82.8$ & $85.4$ \\ \hline
$0$  & $1e-7$ & $124.2$ & $126.3$  \\ \hline
$1$  & $5e-4$ & $17.5$ & $21.2$ \\ \hline
$1$  & $4e-6$ & $25.4$ & $27.7$   \\ \hline
$1$  & $2e-7$ & $38.1$ & $37.6$   \\ \hline
\end{tabular}
\caption{Comparison of $E_{i}[N]$ obtained via analysis and simulation}
\label{tab:ch6:comp_edd_analysis_simln}
\end{center}
\vspace{-1 cm}
\end{table}
A modification of our test is to take into account the available information about the number of samples under $H_0$ (which is not dependent upon $P_1$ in our test) and the fact that the expected drift under $H_1$ is greater than that under $H_0$ if $P_1 \in \mathcal{C}$, i.e., $E_1[N]$ is smaller than $E_0[N]$. Under $H_0$, if the universal estimation is proper we have $N\sim N^0=|\log\beta| /(\lambda/2)$ with high probability. In the ideal case if $\alpha$ is same as $\beta$, we can add the following criteria into the test: decide $H_0$ if the current number of samples $n$ is greater than $N^0$; if $N$ is much smaller than $N^0$ and the decision rule decides $H_0$ we can confirm that it is a miss-detection and make the test not to stop at that point. This improvement will reduce the probability of miss-detection and the mean sample size. In order to improve the above test further if we allow estimation error $\epsilon_n$ at time $n$ ($\epsilon_n$ can be calculated if we know the pointwise redundancy rate of the universal code) the test becomes as provided in Table \ref{tab:ch6:mod_finite_alph}: 

 \begin{table}[h]
 \begin{center}
 \begin{tabular}[t]{|c|c|c|}\hline
  & Stopping rule & Decision \\ \hline
  1 & $n>N^0+\epsilon_n$ & $H_0$ \\ \hline
  2 & $N<<N^0-\epsilon_N$ and declare $H_0$ &  Miss-detection. Do not stop the test \\ \hline
 3 & $N^0-\epsilon_N \leq N \leq N^0+\epsilon_N$ & Decide according to crossing thresholds \\ \hline
 \end{tabular}
\caption{Modified test for finite alphabet case}
 \label{tab:ch6:mod_finite_alph}
\end{center}
\vspace{-1 cm}
 \end{table}
Table \ref{tab:ch6:modfd_alg} shows the performance comparison of the modified test with that of (\ref{eq:ch6:two_sided_finite_alphabet_LLR}). The setup is same as in Table \ref{tab:ch6:comp_edd_analysis_simln} with $\lambda=2.5754$. Since the approximations for $N^0$ hold only when the probability of error is very low, we are interested in low error regime. 

\begin{table}[!tbh]
\begin{center}
\begin{tabular}[t]{|c|c|c|c|}\hline
$Hyp=i$ & {$P_{i}(H_j),j \neq i$} & {$E_{i}[N] \text{ Original}$} & {$E_{i}[N]\text{ Modified}$} \\ \hline
$0$  & $4e-3$ & $18.82$ & $14.27$ \\ \hline
$0$  & $2e-4$ & $27.53$ & $21.92$ \\ \hline
$0$  & $1e-7$ & $55.24$ & $46.15$  \\ \hline
$1$  & $6e-3$ & $15.72$ & $12.08$   \\ \hline
$1$  & $3e-4$ & $25.52$ & $18.98$   \\ \hline
$1$  & $2e-7$ & $50.31$ & $39.12$   \\ \hline

\end{tabular}
\caption{Comparison of $E_{i}[N]$ between the modified test and original test (\ref{eq:ch6:two_sided_finite_alphabet_LLR}).}
\label{tab:ch6:modfd_alg}
\end{center}
\vspace{-1 cm}
\end{table}

\section{Continuous Alphabet}
\label{sec:ch6:uni_seq_con_alp}
The test developed in Section \ref{sec:ch6:uni_seq_fin_alp} can be extended to continuous alphabet sources. Now,  in (\ref{eq:ch6:sprt_LLR}) $P_i$ is replaced by $f_i$, $i=0,1$. Since we do not know $f_1$, we would need an estimate of $Z_n \stackrel{\Delta}=\sum_{k=1}^{n}\log f_1(X_k)$. If $E[\log f_1(X_1)] < \infty$, then by strong law of large numbers, $Z_n/n$ is a.s.\ close to $E[\log f_1(X_1)]$ for all large $n$. Thus, if we have an estimate of $E[\log f_1(X_1)]$ we will be able to replace $Z_n$ as in (\ref{eq:ch6:two_sided_finite_alphabet_LLR}). In the following we get a universal estimate of $E[\log f_1(X_1)]\stackrel{\Delta}=-h(X_1)$, where $h$ is the differential entropy of $X_1$, via the universal data compression algorithms. 

First we quantize $X_i$ via a uniform quantizer with a quantization step $\Delta>0$. Let the quantized observations be $X_i^{\Delta}$ and the quantized vector $X_1^{\Delta},\ldots,X_n^{\Delta}$ be $X_{1:n}^{\Delta}$. We know that $H(X_1^\Delta)+\log \Delta \to h(X_1)$ as $\Delta\to 0$ (\cite{Cover_EIT_book}). Given i.i.d. observations $X_1^\Delta,X_2^\Delta,\ldots,X_n^\Delta$, its code length for a good universal lossless coding algorithm approximates $nH(X_1^\Delta)$ as $n$ increases. This idea gives rise to the following modification to (\ref{eq:ch6:two_sided_finite_alphabet_LLR}), 
\begin{equation}
\label{eq:ch6:two_sided_cont_alphabet_LLR}
\widetilde{W}_n=-L_n(X_{1:n}^{\Delta})-n\log \Delta-\sum_{k=1}^{n}\log f_0(X_k)-n\frac{\lambda}{2}
\end{equation}
and as for the finite alphabet case, to get some performance guarantee, we restrict $f_1$ to a class of densities,
\begin{equation}
\label{eq:ch6:div_condn_cont}
\mathcal{C}=\{f_1:D(f_1||f_0)\geq \lambda\}.
\end{equation}

Let the divergence after quantization be $D(f_1^{\Delta}||f_0^{\Delta})$, $f_i^{\Delta}$ being the probability mass function after quantizing $f_i$. Then by data-processing inequality (\cite{Cover_EIT_book}) $D(f_1||f_0) \geq D(f_1^{\Delta}||f_0^{\Delta})$. When $\Delta \to 0$ the lower bound is asymptotically tight and this suggests choosing $\lambda$ based on the divergence between the continuous distributions before quantization.

The following comments justify the above quantization.
\begin{enumerate}
\item It is known that uniform scalar quantization with variable-length coding of $n$ successive quantizer outputs achieves the optimal operational distortion rate function for quantization at high rates (\cite{Gray_TIT1998}). 
\item We can also consider an adaptive uniform quantizer, which is changing at each time step (\cite{Wang_FTCIT_book}). But this makes the scalar quantized observations dependent (due to learning from the available data at that time) and non-identically distributed. Due to this the universal codelength function is unable to learn the underlying distribution.
\item If we have non-uniform partitions with width $\Delta_j$ at $j^{th}$ bin with probability mass $p_j$, then the likelihood sum in (\ref{eq:ch6:two_sided_cont_alphabet_LLR}) becomes,
\[ -L_n(X_{1:n}^{\Delta})-n\sum_{j} p_j \log \Delta_j-\log f_0(X_1^n)-n\frac{\lambda}{2}.\]
Thus non-uniform quantizers require knowledge of $p_j$ which is not available under $H_1$.
\item Assuming we have i.i.d observations, uniform quantization has another advantage: (\ref{eq:ch6:two_sided_cont_alphabet_LLR}) can be written as 
\[-L_n(X_{1:n}^{\Delta})-\sum_{k=1}^{n}\log (f_0(X_k)\Delta)-n\frac{\lambda}{2}.\]
Under the high rate assumption, $f_0(X_k)\Delta \approx f_i^{\Delta}(X_k^{\Delta})$. Thus, $\widetilde{W}_n$ depends upon the quantized observations only and we do not need to store the original observations.

\item The range of the quantization can be fixed by considering only those $f_1$'s whose tail probabilities are less than a small specific value at a fixed boundary and use these boundaries as range.
\end{enumerate}

We could possibly approximate differential entropy $h(X_1)$ by universal lossy coding algorithms (\cite{Dron_DCC2010}, \cite{Yang_TIT1997}). But these algorithms require a large number of samples (more than 1000) to provide a reasonable approximation. In our application we are interested in minimising the expected number of samples in a sequential setup. Thus, we found the algorithms in \cite{Dron_DCC2010} and \cite{Yang_TIT1997} inappropriate for our applications. 

\section{Universal Source Codes}
\label{sec:ch6:lzslrt-ktslrt}
In this section we present two universal source codes which we will use in our algorithms.
\subsection{LZSLRT (Lempel-Ziv Sequential Likelihood Ratio Test)}
\label{subsec:ch5:LZSLRT}
In the following in (\ref{eq:ch6:two_sided_cont_alphabet_LLR}) we use LZ78 (\cite{Ziv_TIT1978}), which is a well known efficient universal source coding algorithm. We call the resulting test as LZSLRT. LZ78 can be summarized as follows:
\begin{enumerate}
\item Parse the input string into phrases where each phrase is the shortest phrase not seen earlier.
\item Encode each phrase by giving the location of the prefix of
the phrase and the value of the latest symbol in the phrase.
\end{enumerate}

Let $t$ be the number of phrases after parsing $X_1^{\Delta},\ldots,X_n^{\Delta}$ in LZ78 encoder and $|A|$ be the alphabet size of the quantized alphabet. The codelength for LZ78 is
\begin{equation}
\label{eq:ch6:LZ78_codelength}
L_n(X_{1:n}^{\Delta})=\sum_{i=1}^{t}\lceil \log i|A| \rceil.
\end{equation}

At low $n$, which is of interest in sequential detection, the approximation for the log likelihood function via LZSLRT, using (\ref{eq:ch6:LZ78_codelength}) is usually poor as universal coding requires a few samples to learn the source. Hence we add a correction term $n\epsilon_n$, in the likelihood sum in (\ref{eq:ch6:two_sided_cont_alphabet_LLR}), where $\epsilon_n$ is the redundancy for universal lossless codelength function. It is shown in \cite{Kieffer_LNSC_unpublished}, that
\begin{equation}
L_n(X_{1:n}^{\Delta})\leq n\tilde{H}_n(X_{1:n}^\Delta)+n\epsilon_n,
\end{equation}
where
\begin{equation*}
\epsilon_n=C \left(\frac{1}{\log n}+\frac{\log \log n}{n}+\frac{\log \log n}{\log n}\right).
\end{equation*}
Here $C$ is a constant which depends on the size of the quantized alphabet and $\tilde{H}_n(X_1^\Delta)$ is the empirical entropy, which is the entropy calculated using the empirical distribution of samples upto time $n$. Thus the test statistic $\widetilde{W}_n^{LZ}$, is
\begin{equation}
\label{eq:sec_lzslrt_final_expn}
\widetilde{W}_n^{LZ}=-\sum_{i=1}^{t}\lceil \log i|A| \rceil-C \left(\frac{1}{\log n}+\frac{\log \log n}{n}+\frac{\log \log n}{\log n}\right)-n\log \Delta-\sum_{k=1}^{n}\log f_0(X_k)-n\frac{\lambda}{2}.
\end{equation} 
\noindent To obtain $t$, the sequence $X^{\Delta}_{1:n}$ needs to be parsed through the LZ78 encoder.
\vspace{-0.15 cm}
\subsection{KTSLRT (Krichevsky-Trofimov Sequential Likelihood Ratio Test)}
In this section we propose KTSLRT for i.i.d. sources. The codelength function $L_n$ in (\ref{eq:ch6:two_sided_cont_alphabet_LLR}) now comes from the combined use of KT (Krichevsky-Trofimov \cite{Krichevsky_TIT1981}) estimator of the distribution of the quantized source and the Arithmetic Encoder (\cite{Cover_EIT_book}) (i.e., Arithmetic Encoder needs the distribution of $X_n$ which we obtain in this test from the KT-estimator). Together these form a universal source encoder which we call KT-AE. It is proved in \cite{Csiszar_ITS_book} that universal codes defined by the KT-AE are nearly optimal for i.i.d.\ finite alphabet sources. We will show that the test obtained via KT-AE often substantially outperforms LZSLRT.

KT-estimator for a finite alphabet source is defined as,
\begin{equation}
\label{eq:KT-estimator}
P_c(x_1^n)=\prod_{t=1}^{n}\frac{v(x_t/x_{1}^{t-1})+\frac{1}{2}}{t-1+\frac{|A|}{2}},
\end{equation}
where $v(i/x_{1}^{t-1})$ denotes the number of occurrences of the symbol $i$ in $x_1^{t-1}$. It is known (\cite{Cover_EIT_book}) that the coding redundancy of the Arithmetic Encoder is smaller than 2 bits, i.e., if $P_c(x_1^n)$ is the coding distribution used in the Arithmetic Encoder then $L_n(x_1^n)<-\log P_c(x_1^n)+2$. In our test we actually use $-\log P_c(x_1^n)+2$ as the code length function and do not need to implement the Arithmetic Encoder. This is an advantage over the scheme LZSLRT presented above.

Writing (\ref{eq:KT-estimator}) recursively, (\ref{eq:ch6:two_sided_cont_alphabet_LLR}) can be modified as,

\begin{equation}
\hspace{0cm}
\label{eq:ch6:two_sided_cont_alphabet_LLR_KT}
\widetilde{W}_n^{KT}=\widetilde{W}_{n-1}^{KT}+\log \left(\frac{v(X_n^{\Delta}/X_{1}^{\Delta n-1})+\frac{1}{2}+S}{t-1+\frac{|A|}{2}}\right)-\log \Delta-\log f_0(X_n)-\frac{\lambda}{2},
\end{equation}
\normalsize
where $S$ is a scalar constant whose value greatly influences the performance. The default value of $S$ is zero.
\section{Performance Comparison}
\label{sec:ch6:uni_seq_per_com}
We compare the performance of LZSLRT to that of SPRT and a nearly optimal sequential GLR test, GLR-Lai (\cite{Lai_SS2001}), via simulations in Section \ref{subsec:ch6:LZSLRT_simulation}. Performance of KTSLRT through simulations and comparison with LZSLRT are provided in Section \ref{subsec:ch6:KTSLRT_simulation}. We also compare with some other estimators available in literature. It has been observed from our experiments that due to the difference in the expected drift of likelihood ratio process under $H_1$ and $H_0$, some algorithms perform better under one hypothesis and worse under the other hypothesis. Hence instead of plotting $E_{1}[N]$ versus $P_{MD}$ and $E_{0}[N]$ versus $P_{FA}$ separately, we plot $E_{DD} \stackrel{\Delta}{=}\allowbreak 0.5E_{1}[N]+0.5E_{0}[N]$ versus $P_E \stackrel{\Delta}{=}\allowbreak 0.5P_{FA}+0.5P_{MD}$. We use an eight bit uniform quantizer.

\subsection{LZSLRT}
\label{subsec:ch6:LZSLRT_simulation}
\begin{figure}[!htb]
\begin{center}
\includegraphics[trim = 16mm 0 19mm 5mm, clip=true,scale=0.4]{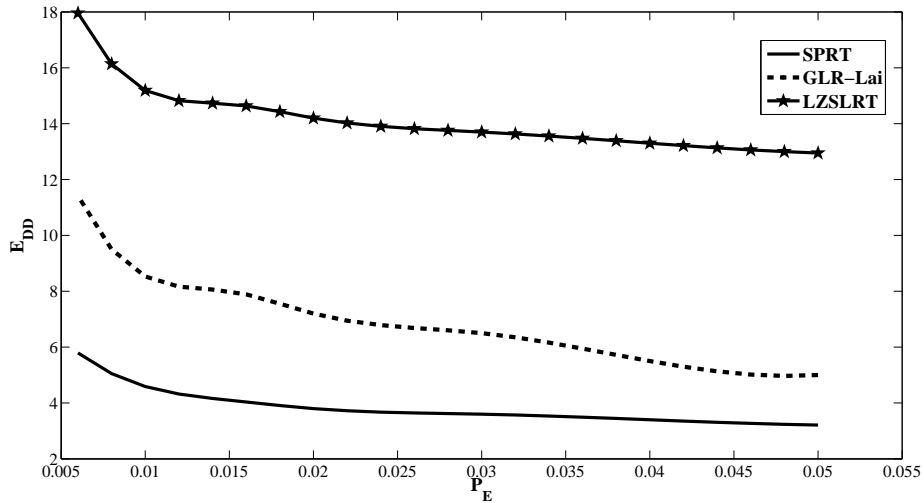}
\caption{Comparison among SPRT, GLR-Lai and LZSLRT for Gaussian Distrbution}
\label{fig:ch6:LZSLRT_single_Gaussian}
\end{center}
\vspace{-1 cm}
\end{figure}

\begin{figure}[!htb]
\begin{center}
\includegraphics[trim = 16mm 0 19mm 5mm, clip=true,scale=0.4]{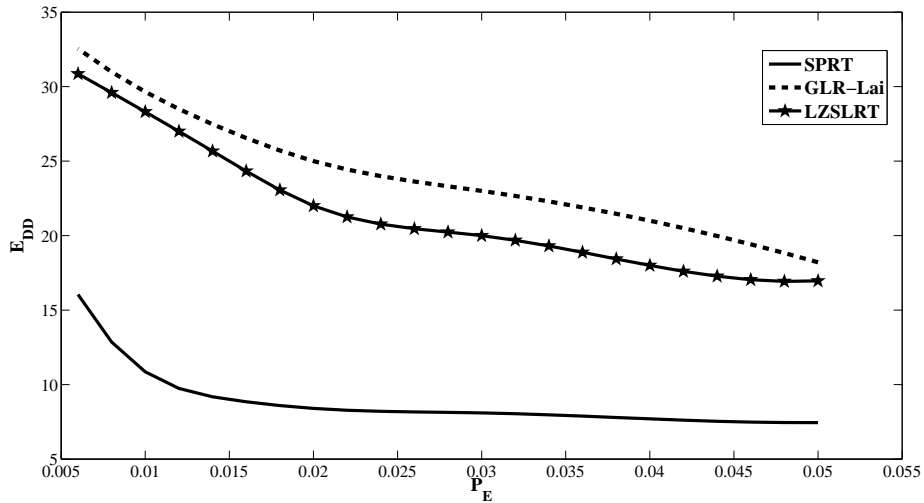}
\caption{Comparison among SPRT, GLR-Lai and LZSLRT for Pareto Distrbution.}
\label{fig:ch6:LZSLRT_single_Pareto}
\end{center}
\vspace{-1 cm}
\end{figure}
%
Figure \ref{fig:ch6:LZSLRT_single_Gaussian} and Figure \ref{fig:ch6:LZSLRT_single_Pareto} present numerical comparisons for Gaussian and Pareto distributions respectively. The experimental set up for Figure \ref{fig:ch6:LZSLRT_single_Gaussian} is, $f_{0}\sim \mathcal{N}(0,5)$, $f_{1}\sim \mathcal{N}(3,5)$ and $\Delta=0.3125$. The setup for Figure \ref{fig:ch6:LZSLRT_single_Pareto} is, $f_{0}\sim \mathcal{P}(10,2)$ and $f_{1}\sim \mathcal{P}(3,2)$, where $\mathcal{P}(K,A)$ is the Pareto density function with $K$ and $A$ as the shape and scale parameter of the distribution. We observe that although LZSLRT performs worse for Gaussian distribution (GLR-Lai is nearly optimal for exponential family), it works better than GLR-Lai for the Pareto Distribution.
\vspace{-0.1 cm}
\subsection{KTSLRT}
\label{subsec:ch6:KTSLRT_simulation}
Figure \ref{fig:ch6:gaussian_var_change} shows the comparison of LZSLRT with KTSLRT when $f_1 \sim \mathcal{N}(0,5)$ and $f_0 \sim \mathcal{N}(0,1)$. We observe that LZSLRT and KTSLRT with $S=0$ (the default case) are not able to give $P_E$ less than 0.3 and 0.23 respectively, although KTSLRT with $S=1$ provides much better performance. We have found in our simulations with other data also that KTSLRT with $S=0$ performs much worse than with $S=1$ and hence in the following we consider KTSLRT with $S=1$ only.
\begin{figure}[!htb]
\begin{center}
\includegraphics[trim = 16mm 0 19mm 5mm, clip=true,scale=0.5]{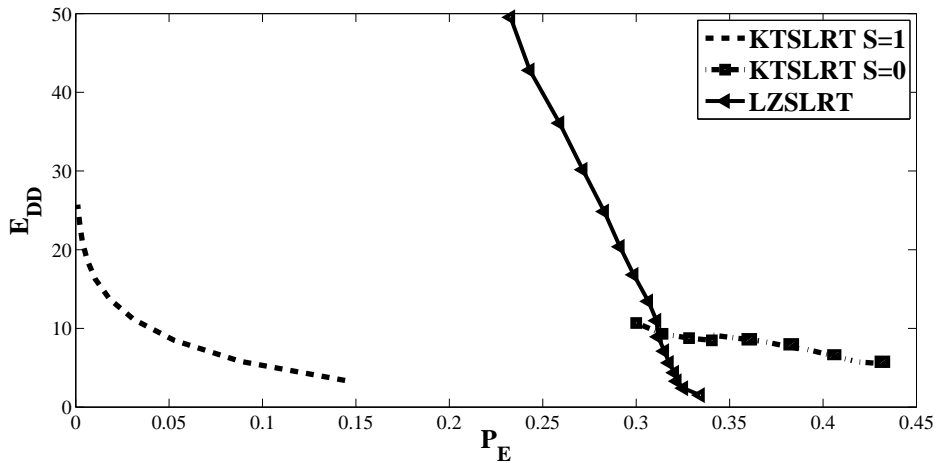}
\caption{Comparison between KTSLRT and LZSLRT for Gaussian Distribution.}
\label{fig:ch6:gaussian_var_change}
\end{center}
\vspace{-1 cm}
\end{figure}
Next we provide comparison for two heavy tail distributions.

Figure \ref{fig:ch6:lognormal} displays the Lognormal distribution comparison when $f_1 \sim ln\mathcal{N}(3,3)$, $f_0 \sim ln\mathcal{N}(0,3)$ and $ln\mathcal{N}(a,b)$ indicates the density function of Lognormal distribution with the underlying Gaussian distribution $\mathcal{N}(a,b)$. It can be observed that $P_E$ less than 0.1 is not achievable by LZSLRT. KTSLRT with $S=1$ provides a good performance.
\begin{figure}[!htb]
\begin{center}
\includegraphics[trim = 16mm 0 19mm 5mm, clip=true,scale=0.5]{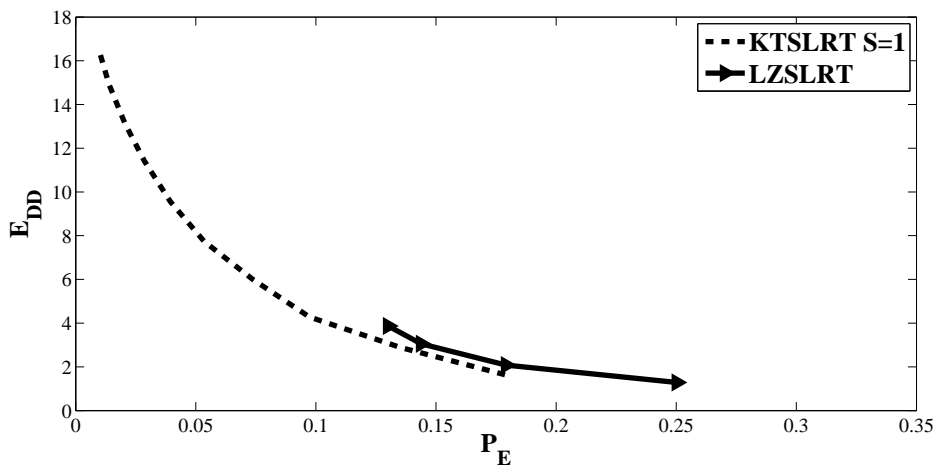}
\caption{Comparison between KTSLRT and LZSLRT for Lognormal Distribution.}
\label{fig:ch6:lognormal}
\end{center}
\vspace{-1 cm}
\end{figure}

Figure \ref{fig:ch6:pareto} shows the results for Pareto distribution. Here $f_1 \sim \mathcal{P}(3,2)$, $f_0 \sim \mathcal{P}(10,2)$ and support set $(2,10)$. We observe that KTSLRT with $S=1$ and LZSLRT have comparable performance. 
\begin{figure}[!htb]
\begin{center}
\includegraphics[trim = 16mm 0 19mm 5mm, clip=true,scale=0.5]{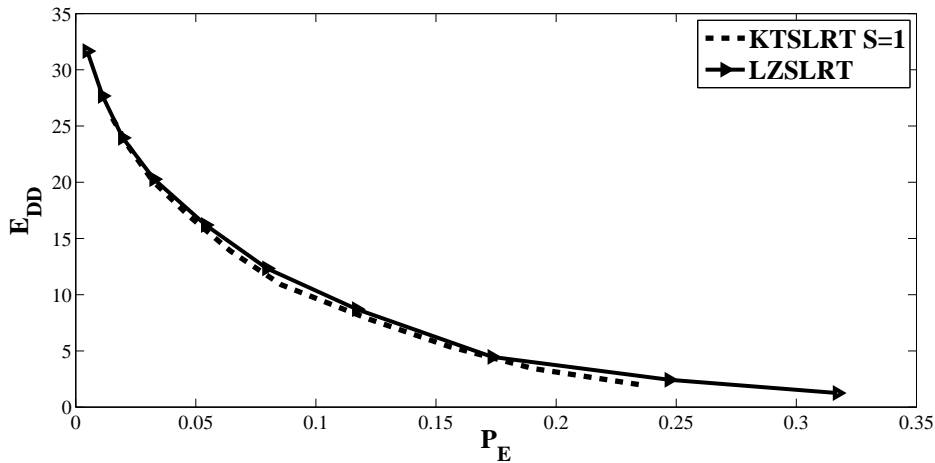}
\caption{Comparison between KTSLRT and LZSLRT for Pareto Distribution.}
\label{fig:ch6:pareto}
\end{center}
\vspace{-1 cm}
\end{figure}

It is observed by us that as $S$ increases, till a particular value the performance of KTSLRT improves and afterwards it starts to deteriorate. For all the examples we considered, $S=1$ provides good performance.

\begin{Remarks}
\label{rem:sup_perf_KTSLRT_arg}
The superior performance of KTSLRT over LZSLRT attributes to the pointwise redundancy rate $n^{-1} (L_n(X_1^n)+\log P(X_1^n))=\mathcal{O}(\log n/n)$ of KT-AE (\cite{Xie_TIT2000}) as compared to $\mathcal{O}(1/\log n)$ of LZ78 (\cite{Kieffer_DCC1999}).
\end{Remarks}

In Figure \ref{fig:ch6:entropy_density_estr} we compare KTSLRT with sequential tests in which $-n\hat{h}_n$ replaces $\sum_{k=1}^{n}\log f_1(X_k)$ where $\hat{h}_n$ is an estimate of the differential entropy and with a test defined by replacing $f_1$ by a density estimator $\hat{f}_n$. 

\begin{figure}[!htb]
\begin{center}
\includegraphics[trim = 16mm 0 19mm 5mm, clip=true,scale=0.5]{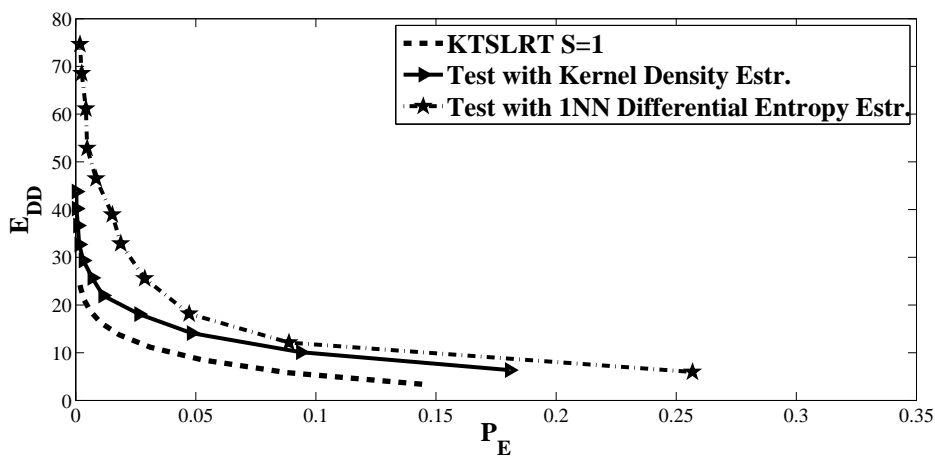}
\caption{Comparison among KTSLRT, universal sequential tests using 1NN differential entropy estimator and that using Kernel density estimator.}
\label{fig:ch6:entropy_density_estr}
\end{center}
\vspace{-1 cm}
\end{figure}

It is shown in \cite{Wang_FTCIT_book} that 1NN (1st Nearest Neighbourhood) differential entropy estimator performs better than other differential entropy estimators where 1-NN differential entropy estimator is 
\[\hat{h}_n=\frac{1}{n}\sum_{i=1}^{n}\log\rho(i)+\log(n-1)+\gamma+1,\]
$\rho(i)\stackrel{\Delta}=\min_{j:1 \leq j\leq n, j\neq i} ||X_i-X_j||$ and $\gamma$ is the Euler-Mascheroni constant (=0.5772...). 

There are many density estimators available (\cite{Silverman_DESDA_book}). We use the Gaussian example in Figure \ref{fig:ch6:gaussian_var_change} for comparison. For Gaussian distributions, a Kernel density estimator is a good choice as optimal expressions are available for the parameters in the Kernel density estimators (\cite{Silverman_DESDA_book}). The Kernel density estimator at a point $z$ is
\[\hat{f}_n(z)=\frac{1}{nh}\sum_{i=1}^{n}K\left(\frac{z-X_i}{h}\right),\]where $K$ is the kernel and $h>0$ is a smoothing parameter called the bandwidth. If Gaussian kernel is used and the underlying density being estimated is Gaussian then it can be shown that the optimal choice for $h$ is (\cite{Silverman_DESDA_book}) ${\left({4{\hat{\sigma}}^5}/{3n}\right)}^{{1}/{5}}$, where $\hat{\sigma}$ is the standard deviation of the samples.

We provide the comparison of KTSLRT with the above two schemes in Figure \ref{fig:ch6:entropy_density_estr}. We find that KTSLRT with $S=1$ performs the best.

Next we provide comparison with the asymptotically optimal universal fixed sample size test for finite alphabet sources. This test is called Hoeffding test (\cite{Hoeffding_AMS1965}, \cite{Levitan_TIT2002}, \cite{Unnikrishnan_TIT2011}) and it is optimal in terms of error exponents (\ref{eq:ch6:FSS_error_exp_opt_crit}) for i.i.d.\,sources over a finite alphabet. The decision rule of Hoeffding test, $\delta_{FSS}=\mathbb{I}\{D(\Gamma^n || P_0) \geq\, \eta\}$, where $\Gamma^n(x)$ is the type of $X_1,\ldots,X_n$, $=\{\frac{1}{n}\sum_{i=1}^{N}\mathbb{I}\{X_i=x\}$, $x\in \mathcal{X}\}$, $\mathcal{X}$ is the source alphabet, $N$ is the cardinality of $\mathcal{X}$ and $\eta >0$ is an appropriate threshold. From \cite[Theorem III.2]{Unnikrishnan_TIT2011},
{\allowdisplaybreaks
\begin{eqnarray}
\text{ under $P_0$, }& &nD(\Gamma^n || P_0) \xrightarrow[n \to \infty]{d} \frac{1}{2}\chi_{N-1}^2, \nonumber \\ 
\text{ under $P_1$, }& & \sqrt{n} \big(D(\Gamma^n || P_0)-D(P_1 || P_0)\big)\xrightarrow[n \to \infty]{d} \mathcal{N}(0,\sigma_1^2),
\end{eqnarray}}
\noindent where $\sigma_1^2=Var_{P_1}\left[\log \frac{P_1(X_1)}{P_0(X_1)}\right]$ and $\chi_{N-1}^2$ is the Chi-Squared distribution with $N-1$ degrees of freedom. From the above two approximations, number of samples, $n$ to achieve $P_{FA}$ and $P_{MD}$ can be computed theoretically as a solution of 
\[2nD(P_1||P_0)+2\sqrt{n} F_{\mathcal{N}}^{-1}(P_{MD})-F_{\chi}^{-1}(1-P_{FA})=0 \,,\]
\noindent where $F_{\mathcal{N}}^{-1}$ and $F_{\chi}^{-1}$ denote inverse cdf's of the above Gaussian and Chi-Squared distributions.

Since this is a discrete alphabet case, we use (\ref{eq:ch6:two_sided_finite_alphabet_LLR}) with $L_n(X_1^n)$ as the codelength function of the universal code, KT-AE. Figure \ref{fig:ch6:comp_hoeff_kt_bern} shows comparison of this test with the Hoeffding test. Here $P_0\sim Be(0.2)$ and $P_1 \sim Be(0.5)$ and $Be(p)$ indicates the Bernoulli distribution with parameter $p$. Figure \ref{fig:ch6:comp_hoeff_kt_binom} provides the comparison when $P_0\sim B(8,0.2)$ and $P_1\sim B(8,0.5)$, where $B(n,p)$ represents the Binomial distribution with $n$ trials and $p$ as the success probability in each trial. It can be seen that our test outperforms Hoeffding test in both these examples in terms of average number of samples.
\begin{figure}[!htb]
\begin{center}
\includegraphics[trim = 16mm 0 19mm 5mm, clip=true,scale=0.4]{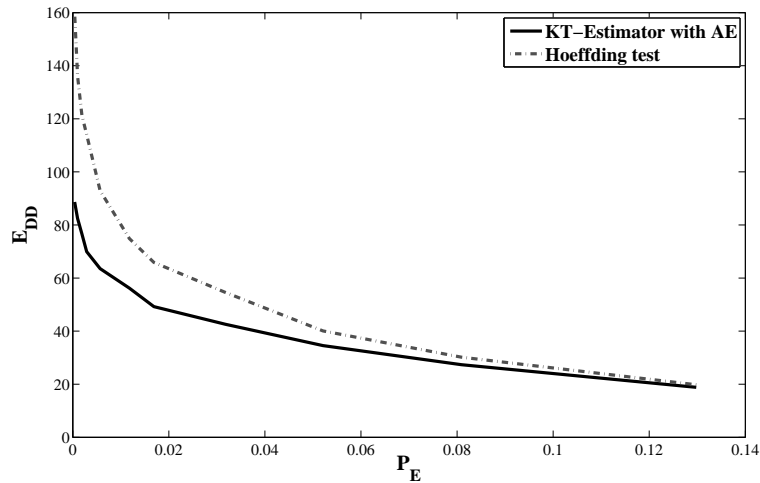}
\caption{Comparison between Hoeffding test and our discrete alphabet test (\ref{eq:ch6:two_sided_finite_alphabet_LLR}) for Bernoulli distribution}
\label{fig:ch6:comp_hoeff_kt_bern}
\end{center}
\vspace{-1 cm}
\end{figure}

\begin{figure}[!htb]
\begin{center}
\includegraphics[trim = 16mm 0 19mm 5mm, clip=true,scale=0.5]{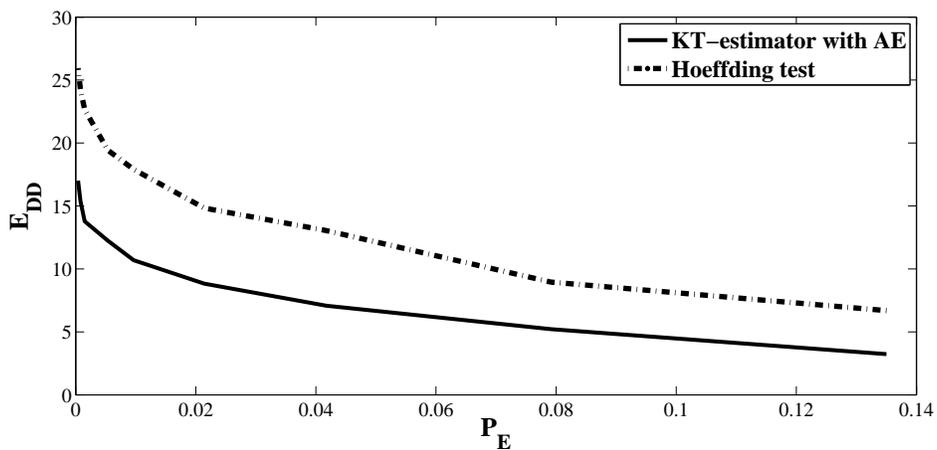}
\caption{Comparison between Hoeffding test and our discrete alphabet test (\ref{eq:ch6:two_sided_finite_alphabet_LLR}) for Binomial distribution}
\label{fig:ch6:comp_hoeff_kt_binom}
\end{center}
\vspace{-1 cm}
\end{figure}

\section{Decentralized Detection}
\label{sec:ch6:uni_seq_dec_det}
\subsection{Algorithm}
\label{subsec:ch6:uni_seq_dec_det_alg}
Motivated by the satisfactory performance of a single node case, we extend LZSLRT and KTSLRT to the decentralized setup in \cite{Banerjee_WCOM} and \cite{Jithin_JSAC-submitted_arxiv}. In this setup we consider a decentralized network with one fusion center (FC) and $L$ local nodes. The local nodes use observations to make local decisions about the presence of a primary and transmit them to the FC. The FC makes the final decision based on the local decisions it received.

Let $X_{k,l}$ be the observation made at local node $l$ at time $k$. We assume that $\{X_{k,l}, k\geq 1\}$ are i.i.d.\ and that the observations are independent across local nodes. We will denote by $f_{1,l}$ and $f_{0,l}$ the densities of $X_{k,l}$ under $H_1$ and $H_0$ respectively. Using the detection algorithm based on $\{X_{n,l},n\leq k\}$ the local node $l$ transmits $Y_{k,l}$ to the fusion node at time $k$. We assume a multiple-access channel (MAC) between the local nodes and the FC in which the FC receives $Y_k$, a coherent superposition of the local node transmissions: $Y_k=\sum_{l=1}^{L}Y_{k,l}+Z_k$, where $\{Z_k\}$ is i.i.d. noise. FC observes $Y_k$, runs a decision rule and decides upon the hypothesis.

Now our assumptions are that at local nodes, $f_{0,l}$ is known but $f_{1,l}$ is not known. The distribution of $Z_k$ is known to the FC. Thus we use LZSLRT at each local node and Wald's SPRT-like procedure at the fusion center (we call it LZSLRT-SPRT). Similarly we can use KTSLRT at each local node and SPRT-like test at the fusion center and call it KTSLRT-SPRT. In both the cases whenever at a local node, a threshold is crossed, it transmits $Y_{k,l}$, which is $b_1$ if its decision is $H_1$; otherwise $b_0$. For thresholds, we use $-\gamma_{0,l}$ and $\gamma_{1,l}$, at local node $l$, as $\log \beta$ and $-\log \alpha$ in Section \ref{sec:ch6:uni_seq_fin_alp}. At the FC we have a sequential test defined by the stopping rule $\min\{n: \sum_{k=1}^n \log \frac{g_{\mu_1}(Y_k)}{g_{-\mu_0}(Y_k)} \notin (-\beta_0,\beta_1)\}$ where $g_{\mu_1}$ is the density of $Z_k+\mu_1$ and $g_{-\mu_0}$ is the density of $Z_k-\mu_0$, and $\mu_0$ and $\mu_1$ are design parameters. The constants $\beta_1$, $\beta_0$, $\mu_0$ and $\mu_1$ have positive values. When the test statistic crosses $\beta_1$, $H_1$ is declared and when it crosses $-\beta_0$, $H_0$ is declared. At the FC, the Log Likelihood Ratio Process $\{F_k\}$ crosses upper threshold under $H_1$ when a sufficient number of local nodes (denoted by $I$, to be specified appropriately) transmit $b_1$. Thus $\mu_1=b_1I$ and similarly $\mu_0=b_0I$.

Thus the overall decentralized algorithm is
\begin{itemize}
\item[(i)] Node $l$ receives $X_{k,l}$ at time $k \geq 1$ and computes $\widehat{W}_{k+1,l}$ as in \eqref{eq:sec_lzslrt_final_expn} or \eqref{eq:ch6:two_sided_cont_alphabet_LLR_KT}, depending on the test selected.
\item[(ii)] Node $l$ transmits 
\begin{equation*}
Y_{k+1,l}=b_1 \mathbb{I}\{\widehat{W}_{k+1} \geq \gamma_{1,l}\}+ b_0 \mathbb{I}\{\widehat{W}_{k+1} \leq -\gamma_{0,l}\}
\end{equation*}
\item[(iii)] Fusion node receives at time $k+1$
\begin{equation*}
Y_{k+1}=\sum_{l=1}^{L} Y_{k+1,l} +Z_k
\end{equation*}
\item[(iv)] Fusion node computes
\begin{equation*}
F_{k+1}=F_k+\log \frac{g_{\mu_1}(Y_k)}{g_{-\mu_0}(Y_k)}
\end{equation*}
\item[(v)] Fusion node decides $H_0$ if $F_{k+1} \leq -\beta_0$ or $H_1$ if $F_{k+1} \geq \beta_1$.
\end{itemize}

In the following we compare the performance of LZSLRT-SPRT, KTSLRT-SPRT and DualSPRT developed in \cite{Jithin_JSAC-submitted_arxiv,Jithin_WCNC2011} which runs SPRT at local node and FC and hence requires knowledge of $f_{1,l}$ at local node $l$. Asymptotically, DualSPRT is shown to achieve performance close to that of the optimal centralized test, which does not consider fusion center noise. We choose $b_1=1$, $b_0=-1$, $I=2$, $L=5$ and $Z_k\sim \mathcal{N}(0,1)$ and assume same SNR for all the local nodes to reduce the complexity of simulations. We use an eight bit quantizer in all these experiments. In Figure \ref{fig:ch6:dec_gaussian}, $f_{0,l}\sim \mathcal{N}(0,1)$ and $f_{1,l} \sim \mathcal{N}(0,5)$, for $1\leq l \leq L$. The setup for Figure \ref{fig:ch6:dec_pareto} is $f_{0,l}\sim \mathcal{P}(10,2)$ and $f_{0,l}\sim \mathcal{P}(3,2)$, for $1\leq l \leq L$. FC thresholds are chosen appropriately with the available expressions for SPRT. In both the cases KTSLRT-SPRT performs better than LZSLRT-SPRT. It also performs better than DualSPRT for higher values of $P_E$.
\begin{figure}[!htb]
\begin{center}
\includegraphics[trim = 16mm 0 19mm 5mm, clip=true,scale=0.5]{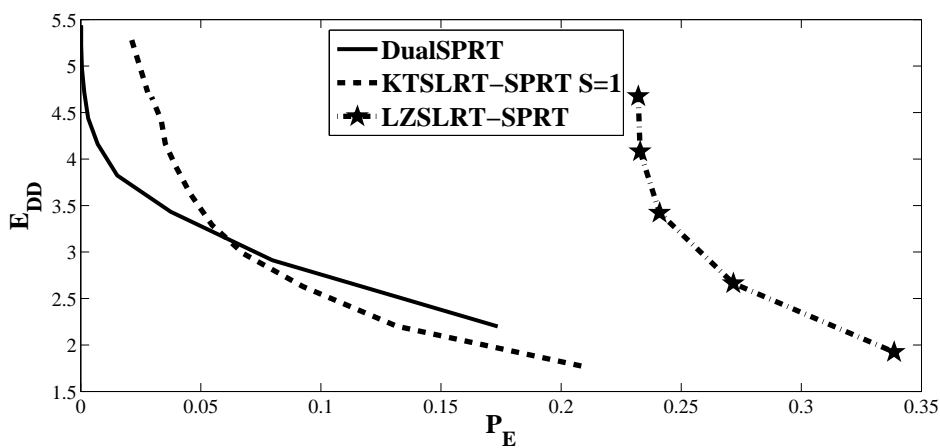}
\caption{Comparison among DualSPRT, KTSLRT-SPRT and LZSLRT-SPRT for Gaussian Distribution}
\label{fig:ch6:dec_gaussian}
\end{center}
\vspace{-1 cm}
\end{figure}

\begin{figure}[!htb]
\begin{center}
\includegraphics[trim = 16mm 0 19mm 5mm, clip=true,scale=0.5]{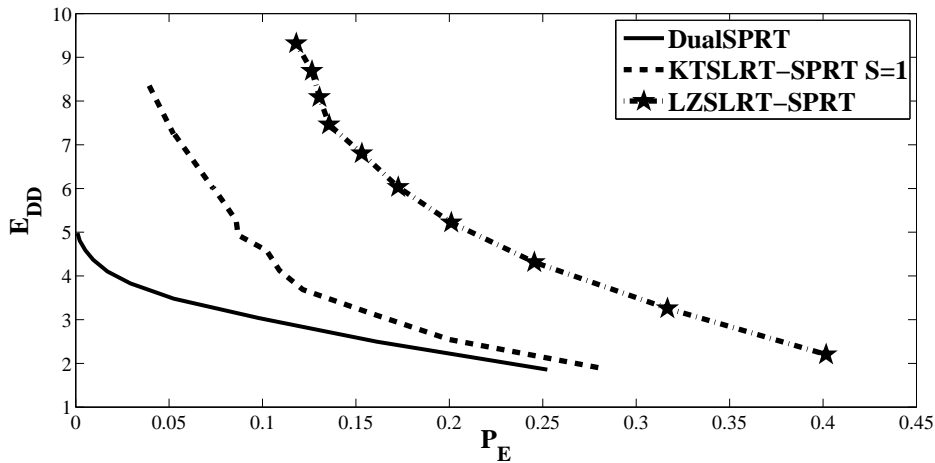}
\caption{Comparison among DualSPRT, KTSLRT-SPRT and LZSLRT-SPRT for Pareto Distribution}
\label{fig:ch6:dec_pareto}
\end{center}
\vspace{-1 cm}
\end{figure}

\subsection{Performance Analysis}
%
Since the analysis is almost same under $H_1$ and $H_0$ with necessary modifications, we provide details only under $H_1$.

At node $l$, let $P_{0,l}$ and $P_{1,l}$ denote the distribution under $H_0$ and $H_1$ respectively, at local node $l$, and 
\[\delta_{l}=E_{1}\left[\log \frac{P_{1,l}(X_{k,l})}{P_{0,l}(X_{k,l})}-\frac{\lambda}{2} \right],\:\rho^2_{l}=Var_{H_1} \left[\log \frac{P_{1,l}(X_{k,l})}{P_{0,l}(X_{k,l})}-\frac{\lambda}{2} \right]. \]
We will assume $\delta_{l}$ finite throughout this paper. By Jensen's Inequality and (\ref{eq:ch6:div_condn_cont}), $\delta_{l} >0$ under $H_1$. Let
\[N_{l}=\inf \{k: \widehat{W}_{k,l} \notin (-\gamma_{0,l},\gamma_{1,l})\},\;N_{l}^1=\inf \{k: \widehat{W}_{k,l} >\gamma_{1,l} \},\,N_{l}^0=\inf \{k: \widehat{W}_{k,l} <-\gamma_{0,l} \}.\]
Then $N_l=\min\{N_l^0,N_l^1\}$. Let $N_{d}$ denote the stopping time $\inf\{k:F_k \notin (-\beta_0, \beta_1)\}$. Let $N^0_{d}=\inf\{k:F_k\leq -\beta_0\}$ and $N^1_{d}=\inf\{k:F_k\geq \beta_1\}$. Then $N_{d}=\min\{N^1_{d},N^0_{d}\}$.

In the rest of this section, we choose $\gamma_{1,l}=\gamma_{0,l}=\gamma$, $\beta_1=\beta_0=\beta$, $\mu_1=\mu_0=\mu$ and $b_1=-b_0=b$ for simplicity of notation.

From Theorem \ref{thm:ch6:edd_prtys}(b), if $(L_n(X_1^n)+\log P_1(X_1^n))/\sqrt{n} \to 0$ a.s.\ ,
\begin{equation}
\label{eq:ch3:N_l_1_pdf}
N_{l}^1\sim\mathcal{N}(\frac{\gamma}{\delta_{l}}, \frac{\rho^2_{l}\, \gamma}{\delta_{l}^3}).
\end{equation}
\begin{Remarks}
From Remark \ref{rem:sup_perf_KTSLRT_arg}, it can be seen that KT-AE satisfies the condition for \eqref{eq:ch3:N_l_1_pdf}, but not LZ78. The following decentralized analysis is applicable for any universal source code which satisfies this condition.
\end{Remarks}
\subsubsection{$E[N_{d}|H_1]$ Analysis}
At the fusion node $F_k$ crosses $\beta$ first under ${H_1}$ with a high probability when a sufficient number of local nodes transmit $b_1$. The dominant event occurs when the number of local nodes transmitting are such that the mean drift of ${F_k}$ will just have turned positive. In the following we find the mean time to this event and then the time to cross $\beta$ after this. 

The following lemmas provide justification for considering only the events $\{N_l^1\}$ and $\{N^1_{d}\}$ for analysis of $E[N_{d}|H_1]$.
\begin{lemma}
\label{lemma:ch3:N_l_N_wp1}
$P_{1}(N_l=N_l^1) \to 1$ as $\gamma \to \infty$ and $P_{1}(N_{d}=N^1_{d}) \to 1$ as $\gamma \to \infty$ and $\beta \to \infty$.
\end{lemma}
\begin{IEEEproof}
We have, $\widehat{W}_{n,l}/n \to D(P_{1,l}||P_{0,l})-\lambda/2$ a.s.\ since \eqref{eq:ch6:pointwise_universality} holds and $\sum_{k=1}^n \log \frac{P_{1,l}(X_{k,l})}{P_{0,l}(X_{k,l})}\to D(P_{1,l}||P_{0,l})$ a.s. Thus by (\ref{eq:ch6:div_condn_cont}), $\widehat{W}_n \to \infty$ a.s. This in turn implies that $\widehat{W}_n$ never crosses some finite negative threshold a.s. This implies that $P_{1}(N_l^0<\infty) \to 0$ as $\gamma \to \infty$ but $P_{1}(N_l^1<\infty)=1$ for any $\gamma < \infty$. Thus $P_{1}(N_l=N_l^1) \to 1$ as $\gamma \to \infty$. This also implies that for large $\gamma$, the drift of $F_k$ is positive for $H_1$ with a high probability and $P_{1}(N_{d}=N^1_{d}) \to 1$ as $\gamma \to \infty$ and $\beta \to \infty$.
\end{IEEEproof}

From Lemma \ref{lemma:ch3:N_l_N_wp1} we also get that under $H_1$, $|N_l-N_l^1| \to 0$ a.s.\ as $\gamma \to \infty$ and $|N_{d}-N^1_{d}| \to 0$ a.s.\ as $\gamma \to \infty $ and $\beta$ $\to \infty$. From this fact, along with Theorem \ref{thm:ch6:edd_prtys}, we can use the result in (\ref{eq:ch3:N_l_1_pdf}) for $N_l$ also. The following lemma also holds.

\begin{lemma}
\label{lemma:ch3:VS_argument_reg_finiteness_st}
Let $t_k$ be the time when $k$ local nodes have made the decision. As $\gamma \to \infty$,
\begin{multline}
\nonumber
P_{1}(\text{Decision at time $t_k$ is $H_1$ and }
\text{$t_k$ is the $k^{th}$ order statistics of $N_1^1,N_2^1,\ldots,N_L^1$}) \to 1.
\end{multline}
\end{lemma}
\begin{proof}
From Lemma \ref{lemma:ch3:N_l_N_wp1}, \vspace{-0.3 cm}
\begin{multline}
\nonumber
P_{1}(\text{Decision at time $t_k$ is $H_1$ and}\text{ $t_k$ is the $k^{th}$ order statistics of $N_1^1,N_2^1,\ldots,N_L^1$})\\
\geq P_{1}(N_l^1<N_l^0, l=1,\ldots,L) \to 1,\, \text{ as } \gamma \to \infty.\qedhere 
\end{multline} 
\end{proof}

We use Lemma \ref{lemma:ch3:N_l_N_wp1}-\ref{lemma:ch3:VS_argument_reg_finiteness_st}, Theorem \ref{thm:ch6:edd_prtys} and equation (\ref{eq:ch3:N_l_1_pdf}) in the following to obtain an approximation for $E[N_{d}|H_1]$ when $\gamma$ and $\beta$ are large. Large $\gamma$ and $\beta$ are needed for small probability of error. Then we can assume that the local nodes are making correct decisions. Let $\delta_{FC}^j$ be the mean drift of $F_k$, when $j$ local nodes are transmitting. Then $t_j$ is the point at which the drift of $F_k$ changes from $\delta_{FC}^{j-1}$ to $\delta_{FC}^j$ and let $\bar{F_j}=E_1[F_{t_j-1}]$, the mean value of $F_k$ just before transition epoch $t_j$. 

Let
\[l^*=min\{j:\delta_{FC}^j>0 \text{ and } \frac{\beta-\bar{F_j}}{\delta_{FC}^j}< E[t_{j+1}]-E[t_j]\}.\]
\noindent $\bar{F_j}$ can be iteratively calculated as
\begin{equation}
\label{eq:ch3:bar{F_j}}
\bar{F_j}=\bar{F}_{j-1}+\delta_{FC}^j\,(E[t_j]-E[t_{j-1}]),\, \bar{F}_0=0.
\end{equation}
Note that $\delta_{FC}^j\, 0 \leq j \leq L$ can be found by assuming $E_1[Yk]$ as $bj$ and $t_j$ as the $j^{\text{th}}$ order statistics of $\{N_l^1,\, 0\leq l \leq L \}$. The Gaussian approximation (\ref{eq:ch3:N_l_1_pdf}) can be used to calculate the expected value of the order statistics using the method given in \cite{Barakat_SMA2004}. This implies that $E[t_j]$s and hence $\bar{F_j}s$ are available offline. By using these values $E_1[N_{d}]$ ($\approx E_1[N^1_{d}]$) can be approximated as,
\begin{equation}
\label{Eddanalysis1}
E_1[N_{d}]\approx E[t_{l^*}]+\frac{\beta-\bar{F_{l^*}}}{\delta^{l^*}_{FC}},
\end{equation}
\noindent where the first term on R.H.S.\ is the mean time till the drift becomes positive at the fusion node while the second term indicates the mean time for $F_k$ to cross $\beta$ from $t_{l^*}$ onward.

In case of continuous alphabet sources which is assumed in our decentralized algorithm, $\widetilde{W}_n$ in (\ref{eq:ch6:two_sided_cont_alphabet_LLR}) can be modified to
\begin{IEEEeqnarray}{rCl}
 & &-L_n(X_{1:n}^\Delta)-\sum_{k=1}^n\log (f_1(X_k)\Delta)+\sum_{k=1}^n \log \frac{f_1(X_k)}{f_0(X_k)}-\frac{\lambda}{2} \nonumber \\
 & \stackrel{(a)}\approx & -L_n(X_{1:n}^\Delta)-\log (f_1^\Delta(X_{1:n}^\Delta))+\sum_{k=1}^n \log \frac{f_1(X_k)}{f_0(X_k)}-\frac{\lambda}{2}. \IEEEeqnarraynumspace \nonumber
 \label{eq:sec_dec_ana:modfd_w_tilde}
 \end{IEEEeqnarray}
Here $f_1^\Delta$ is the probability mass function after quantizing $X_{k,l}$, which is the distribution being learnt by $L_n(X_{1:n}^\Delta)$. Approximation $(a)$ is due to the approximation $f_1(x)\Delta \approx f_1^{\Delta}(x^\Delta)$ at high rate uniform quantization. By taking $\xi_n=-L_n(X_{1:n}^\Delta)-\log (f_1^\Delta(X_{1:n}^\Delta))$ and $S_n=\sum_{k=1}^n \log \frac{f_1(X_k)}{f_0(X_k)}-\frac{\lambda}{2}$, it is clear that $\widetilde{W}_n$ can be approximated as a perturbed random walk since $\{S_n,n\geq1\}$ is a random walk and $\xi_n/n \to 0 $ a.s.\ from the pointwise convergence of universal source codes.
\subsubsection{$P_{MD}$ Analysis}
At reasonably large local node thresholds, according to Lemma \ref{lemma:ch3:VS_argument_reg_finiteness_st}, with a high probability local nodes are making the right decisions and $t_k$ can be taken as the order statistics assuming that all local nodes make the right decisions. $P_{MD}$ at the fusion node is given by,
\begin{equation}
P_{MD}=P_1(\text{accept $H_0$})=P_1(N^0_d<N^1_d). \nonumber
\end{equation}
It can be easily shown that $P_{1}(N^1_d< \infty)=1$ for any $\beta>0$. Also $P_{1}(N^0_d< \infty) \to 0 $ as $\beta \to \infty$. We should decide the different thresholds such that $P_{1}(N^1_d<t_1)$ is small for reasonable performance. Therefore,
\begin{eqnarray}
\label{eq:ch3:pfa_analysis_new_arg}
P_{MD} = P_{1}(N^0_d<N^1_d) & \geq & P_{1}(N^0_d<t_1, N^1_d> t_1) \approx P_{1}(N^0_d <t_1).
\end{eqnarray}
Also,
\begin{IEEEeqnarray}{rCl}
\label{pfa_analysis_ktslrt-sprt}
P_1(N^0_d<N^1_d)&\leq & P_{1}(N^0_d<\infty)= P_{1}(N^0_d<t_1)+ P_{1}(t_1 \leq N^0_d < t_2)+P_{1}(t_2 \leq N^0_d < t_3)+\ldots \IEEEeqnarraynumspace
\end{IEEEeqnarray}
\noindent The first term in the right hand side is expected to be the dominant term. This is because, from Lemma \ref{lemma:ch3:VS_argument_reg_finiteness_st}, after $t_1$, the drift of $F_k$ will be most likely more positive than before $t_1$ (if $P_{MD}$ at local nodes are reasonably small) and causes fewer errors if the fusion center threshold is chosen appropriately. We have verified this from simulations also. Hence we focus on the first term.
Combining this fact with \eqref{eq:ch3:pfa_analysis_new_arg}, $P_1(N^0_d<t_1)$ will be a good approximation for $P_{1}(\text{reject $H_1$})$. For calculating $P_1(N^0_d<t_1)$, we use the bounding technique and approximate expression given in \cite[Section III-B2]{Jithin_JSAC-submitted_arxiv} with the distribution of $N_l^1$ in (\ref{eq:ch3:N_l_1_pdf}).

Figure \ref{fig:comp_sym_analys_dec} provides comparison between analysis and simulations for continuous distributions. The simulation setup is same as in Figure \ref{fig:ch6:dec_gaussian}. It shows that at low $P_{MD}$, $E_1[N_{d}]$ from theory approximates the simulated value reasonably well.
%
\begin{figure}[h]
\vspace{-0.2 cm}
\centering
\subfloat[Probability of miss-detection]{\includegraphics[trim = 16mm 0 19mm 5mm, clip=true,scale=0.4]{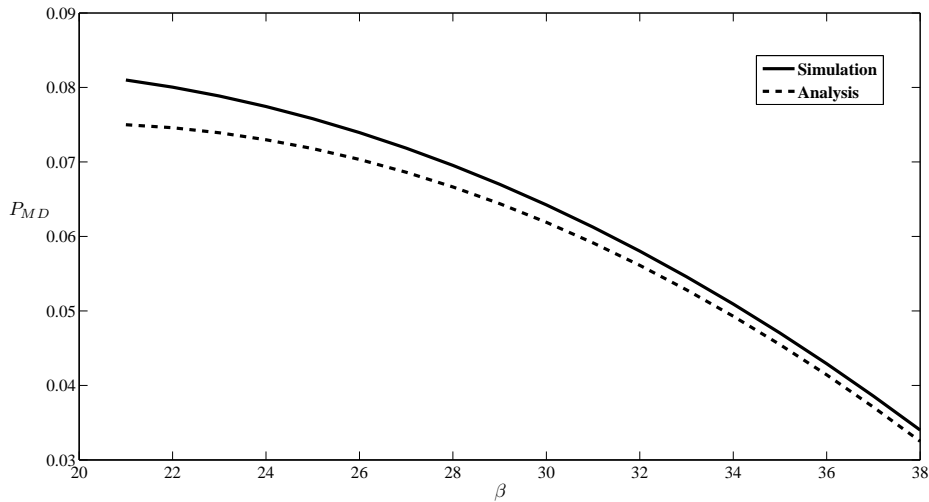}} \\
\subfloat[Expected delay]{\includegraphics[trim = 16mm 0 19mm 5mm, clip=true,scale=0.4]{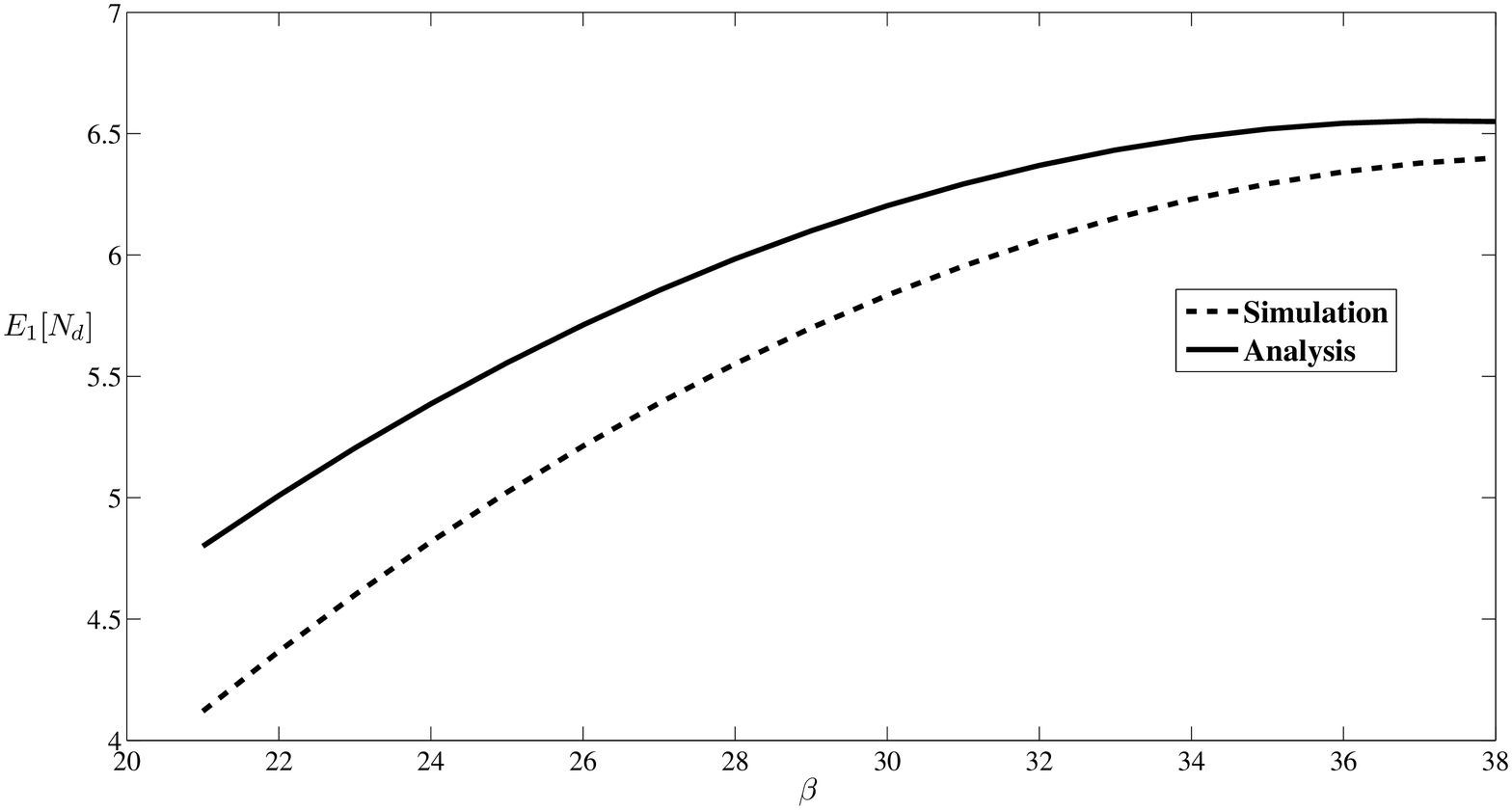}\label{fig_dualSPRT_perf-comp_same-SNR-2}}
\vspace{-0.2 cm} \caption{KTSLRT-SPRT:Comparison of $E_1[N_{d}]$ and $P_{MD}$ obtained via analysis and simulation.} 
\vspace{-0.8 cm}
\label{fig:comp_sym_analys_dec}
\end{figure}

\section{Asymptotic properties of the decentralized test}
\label{sec:DualSPRT_asy_opt}
In this section we prove asymptotic properties of the decentralized test. 

We use the following notation:
\[D_{tot}^0=L\lambda/2,\, D_{tot}^1=\sum_{l=1}^{L} (D(f_{1,l}||f_{0,l})-\lambda/2),\,r_l=\frac{\lambda/2}{D_{tot}^0},\, \rho_l=\frac{D(f_{1,l}||f_{0,l})-\lambda/2}{D_{tot}^1}.\]

Let $\mathcal{A}^i$ be the event that all the local nodes transmit $b_i$ when the true hypothesis is $H_i$. Also let $\Delta(\mathcal{A}^i)$ be the drift of the fusion center process $F_k$ when the $\mathcal{A}^i$ happens, i.e., $\Delta(\mathcal{A}^i)=E_i \left[(\log \frac{g_{\mu_1}(Y_k)}{g_{-\mu_0}(Y_k)})|\mathcal{A}^i\right]$. We use $\theta_i$ as the mean of the increments of $F_k$ when all the local nodes transmit wrong decisions under $H_i$. We will also need
{\allowdisplaybreaks
\begin{eqnarray}
\label{eq:ch3:proof_notn_tau}
\tau_l^*(a)&\stackrel{\Delta}=&\sup \left\{n\geq1:\widehat{W}_{n,l}\geq  a \right\}, \,
\tau^*(a)\stackrel{\Delta}=\max_{1 \leq l \leq L} \tau_l^* (ar_l).
\end{eqnarray}}
It can be seen that $\tau_l^*(a)$ is the last time $\widehat{W}_{n,l}$ will be above $a$ under $H_0$. Let $\tau_l(a)$ be the last exit time of a random walk, at node $l$, with drift $-\lambda/2-\epsilon$ from the interval $(\infty, a)$, for any $\epsilon>0$. $N_{0,l}^*(\epsilon)$ is the $N_{0}^*(\epsilon)$ in \eqref{eq:ch6:N_0(eps)_def} at local node $l$.

Let $F_k^*$ be another log likelihood ratio sum (at FC) with expected value of its components as $\theta_i$ under $H_i$, the worst case value of the mean of the increments of $F_k$. Let the increments of $F_k^*$ be $\xi_1^*,\ldots,\xi_k^*$ which are i.i.d. Under $H_0$, $\theta_0 >0$ and under $H_1$, $\theta_1<0$. 

In the rest of this section, local node thresholds are $\gamma_{0,l}=-r_l|\log c|, \gamma_{1,l}=\rho_l|\log c|$ and fusion center thresholds are $\beta_0=-{|\log c|}, \beta_1={|\log c|}$.

\begin{theorem}
\label{thm:ch3:edd}
Assume the following: for some $\alpha>1$,
\begin{itemize}
\item[(i)] $E[N_{0,l}^*(\epsilon)]< \infty$ for any $\epsilon >0$, for all $l=1,\ldots,L$,
\item[(ii)] $E[{(\log \frac{f_{1,l}(X_{1,l})}{f_{0,l}(X_{1,l})})}^{\alpha+1}]<\infty$ for all $l=1,\ldots,L$,
\item[(iii)] $E[|\xi_1^*|^{\alpha+1}]<\infty$.
\end{itemize}
Let $\rho^2_{i,l}<\infty$ for all $l$. Then under $H_i$, a.s.\ and in expectations
\begin{equation*}
\overline{\lim_{c \to 0}} \frac{N_d}{|\log c|} \leq \frac{1}{D_{tot}^i}+M_i, 
\end{equation*}
where $M_i=C_i/\Delta(\mathcal{A}^i)$, $C_0=-\left[1+\frac{E_0[|\xi_1^*|]}{D_{tot}^0} \right]$ and $C_1=\left[1+\frac{E_1[|\xi_1^*|]}{D_{tot}^1} \right]$.
\end{theorem}

\begin{IEEEproof}
See Appendix \ref{proof:thm:ch3:edd}.
\end{IEEEproof}

$D_{tot}^0$, $D_{tot}^1$ and $\Delta(\mathcal{A}^0)$, $\Delta(\mathcal{A}^1)$ given above indicate the advantage of using more local nodes $(L)$ and having higher drift $\lambda/2$ at the local nodes.

Next we consider the asymptotics of $P_{FA}$ and $P_{MD}$. Let 
\begin{equation*}
R_i=\min_{1\leq l \leq L} \Big(-\log \inf_{t \geq 0} E_i\Big[\exp \left({-t \log \frac{f_{1,l(X_{1,l})}}{f_{0,l}(X_{1,l})}}\right)\Big] \Big),
\end{equation*} 
and $\xi_k=\log \frac{g_{\mu_1}(Y_k)}{g_{-\mu_0}(Y_k)}$. Then $F_n=\sum_{k=1}^n \xi_k$.

Let $G_i$ and $\widehat{G}_i$ be the distributions of $|\xi_1^*|$ and $\xi_1^*$ respectively. Also let $g_i$ and $\widehat{g}_i$ be the moment generating functions of $|\xi_1^*|$ and $\xi_1^*$. Let $\Lambda_i(\alpha)=\sup_{\lambda}(\alpha \lambda-\log g_i(\lambda))$, $\widehat{\Lambda}_i(\alpha)=\sup_{\lambda}(\alpha \lambda-\log \widehat{g}_i(\lambda))$ and take $\alpha^{+}_i=\text{ess }\sup |\xi_1^*|$. Let
\begin{equation}
\label{eq:thm_pe_s-param}
s_i(\eta)=\left\{
\begin{array}{ll}
\frac{\eta}{\alpha_i^{+}} &, \text{ if } \eta \geq \Lambda_i(\alpha^{+}_i),\\
\frac{\eta}{\Lambda^{-1}_i(\eta)} &, \text{ if } \eta \in (0, \Lambda_i(\alpha^{+}_i)).
\end{array}
\right.
\end{equation}
\begin{theorem}
\label{thm:ch3:pe-bayes}
Let $g_i(\lambda)<\infty$ in a neighbourhood of zero. Then,
\begin{enumerate}
\item[(a)]\label{thm:ch3:peb} $\displaystyle \lim_{c \to 0}\, \frac{P_{FA}}{c}=0$ if for some $0<\eta<R_0$, $s_0(\eta) >1$.
\item[(b)]\label{thm:ch3:pec} $\displaystyle \lim_{c \to 0}\, \frac{P_{MD}}{c}=0$ if for some $0<\eta<R_1$, $s_1(\eta) >1$. 
\end{enumerate}
\end{theorem}
\begin{IEEEproof}
See Appendix \ref{proof:thm_pe}.
\end{IEEEproof}
\begin{Remarks}
\label{rem:asy_opt_DualSPRT_rem-contraction-pri}
When $\alpha_i^{+}=\infty$ which is generally true, $\Lambda_i(\alpha^{+}_i)=\infty$ (\cite{Borovkov_1995}) and in Theorem \ref{thm:ch3:pe-bayes}(a) and \ref{thm:ch3:pe-bayes}(b) we need to consider only $R_i<\Lambda_i(\alpha^{+}_i)$. Also $\Lambda_i$ can be computed from $\widehat{\Lambda}_i$ using Contraction principle in Large Deviation theory (\cite{Dembo_LDP}).
\end{Remarks}

\begin{Remarks}
In \cite[Lemma 1-Appendix A]{Banerjee_WCOM}, it is proved that log likelhood ratio converts a large class of distributions into light tailed distributions and hence $g_i(\lambda)$ is finite in a neighbourhood of zero. For instance, consider a regularly varying distribution for $Z_k$, $P(Z_k>x)=l'(x)x^{-\alpha}$, where $l'(x)$ is a slowly varying function and $\alpha>0$. Then,
$\log g_{\mu_1}(x)/g_{-\mu_0}(x)=\Big(l(x-\mu_1)/l(x+\mu_0)\Big) (x-\mu_1)^{-\alpha}(x+\mu_0)^{\alpha}$ $\leq x^{\beta_1+\alpha\beta_2}$ for large $x$, any $\beta_1>0$ and an appropriately chosen $\beta_2>0$. This proves the conditions for \cite[Lemma 1]{Banerjee_WCOM} and hence exponential tail for $\widehat{G}_{i}(t)$ follows.
\end{Remarks}

Theorem \ref{thm:ch3:pe-bayes} provides optimal rate of convergence of $P_{FA}$ and $P_{MD}$ (see, e.g., \cite[Theorem 2.11.2]{Govindarajulu_SS}).

\subsection{Example-Gaussian distribution}
\label{subsec:asy_prop_ex_gau}
In the following we apply Theorems \ref{thm:ch3:edd} and \ref{thm:ch3:pe-bayes} when the fusion center noise is Gaussian $\mathcal{N}(0,\sigma^2_{FC})$. We take $\mu_1=\mu_0=\mu>0$ and $b_1=-b_0=b>0$. For Theorem \ref{thm:ch3:edd}, $\Delta(\mathcal{A}^0)=-2\mu Lb/ \sigma^2_{FC}$ and $\Delta(\mathcal{A}^1)=2\mu Lb/ \sigma^2_{FC}$. Therefore $M_0$ and $M_1$ in Theorem \ref{thm:ch3:edd} $\to 0$ if $L \to \infty$ and/or $b \to \infty$ and we can obtain performance as close to asymptotic optimality as we wish. This also happens if $\sigma^2_{FC} \to 0$.

Using Remark \ref{rem:asy_opt_DualSPRT_rem-contraction-pri}, the condition in Theorem \ref{thm:ch3:pe-bayes}(a) is $\sigma^2_{FC}\eta/(4\mu^2\sqrt{2\eta}-2\mu L b)> 1$ for some $0<\eta<R_0$ and that for Theorem \ref{thm:ch3:pe-bayes}(b) is $\sigma^2_{FC}\eta/(4\mu^2\sqrt{2\eta}+2\mu L b)> 1$ for some $0<\eta<R_1$. Combining these two, it is sufficient to satisfy later condition with $0<\eta<\min(R_0,R_1)$. This specifies upper-bounds for the choice of $\mu,L$ and $b$.
\section{Conlusions}
\label{sec:ch6:uni_seq_con}
The problem of universal sequential hypothesis testing is very useful in practical applications, e.g., quickest detection with SNR uncertainty in Cognitive Radio systems. We have used universal lossless source codes for learning the underlying distribution. The algorithm is first proposed for discrete alphabet and almost sure finiteness of the stopping time is proved. Asymptotic properties of probability of error and stopping times are also derived. Later on the algorithm is extended to continuous alphabet via uniform quantization. We have used Lempel-Ziv code and KT-estimator with Arithmetic Encoder as universal lossless codes. From the performance comparisons, it is found that KT-estimator with Arithmetic Encoder (KT-AE) always performs the best. We have compared this algorithm with other universal hypothesis testing schemes also and found that KT-AE performs the best. Finally we have extended these algorithms to decentralized setup and studied their asymptotic performance.

\appendices
\section{Proof of Proposition \ref{prop:ch6:as_st}}
\label{proof:prop:ch6:as_st}
(a)\quad Since $P_0(N < \infty)\geq P_0(N^0 < \infty)$, we show $P_0(N^0 < \infty)=1$. 

From our assumptions, we have, as $n \to \infty$,
\begin{equation*}
\frac{\widehat{W}_n}{n}={}-\frac{L_n(X_1^n)}{n}-\frac{\log P_0(X_1^n)}{n}-\frac{\lambda}{2} \to -\frac{\lambda}{2} \quad \text{ in probability.}
\end{equation*}

Therefore,
\begin{equation*}
P_0[N^0 < \infty]\geq P_0[\widehat{W}_n< \log \beta]=P_0\left[\frac{\widehat{W}_n}{n} <\frac{\log \beta}{n}\right] \to 1.
\end{equation*}
(b)\quad The proof follows as in (a), observing that $P_1(N < \infty)\geq P_1(N^1 < \infty)$ and $\widehat{W}_n/n \to D(P_1||P_0)-\lambda/2 > 0$ in probability. \hfill \QED

\section{Proof of Theorem \ref{thm:ch6:pfa_pmd_prtys}}
\label{proof:thm:ch6:pfa_pmd_prtys}
(1)\quad We have, 
\[P_{FA}=P_0(N^1<N^0)\leq P_0(N^1<\infty).\]
$P_0(N^1<\infty)\leq \alpha$ is proved in \cite{Tony_ISIT2008} and is provided here for the sake of completeness. It uses the fact that the universal codes we consider are prefix-free and hence satisfy the Kraft's inequality (\cite{Cover_EIT_book}). Thus,
{\allowdisplaybreaks
\begin{eqnarray}
P_0(N^1<\infty)&=&\sum_{n=1}^{\infty} P_0[N^1=n]=\sum_{n=1}^{\infty}P_0\left[-L_n(X_1^n)-\log P_0 (X_1^n) -n\frac{\lambda}{2} \geq -\log \alpha \right]\nonumber\\
&\leq & \sum_{n=1}^{\infty} \sum_{x_1^n:P_0(x_1^n) \leq 2^{\log \alpha -L_n(X_1^n)-n\frac{\lambda}{2}}}P_0(x_1^n) \nonumber \\
& \leq & \sum_{n=1}^{\infty} \sum_{x_1^n:P_0(x_1^n) \leq 2^{\log \alpha -L_n(X_1^n)-n\frac{\lambda}{2}}} \alpha\, 2^{-L_n(X_1^n)-n\frac{\lambda}{2}} \nonumber \\
&\stackrel{(a)}\leq & \sum_{n=1}^{\infty} {\alpha}\,{2^{-n\lambda/2}}=\frac{\alpha}{2^{\lambda/2}-1} \leq \alpha. \nonumber
\end{eqnarray}}
where (a) follows from Kraft's inequality.\\
(2)\quad Let $\mathcal{A}_{n_1}(\epsilon) = \{x_1^{\infty}: \sup_{n \geq n_1} |-L_n(x_1^n)-\log P_1(x_1^n)| < n \epsilon \}$. We have, for any $n_1 >0$,
{\allowdisplaybreaks
\begin{eqnarray}
\label{eq:ch6:pmd_proof_exp1}
P_{MD}=P_1(N^0<N^1)&= &P_1[N^0<N^1;N^0 \leq n_1]+P_1[N^0<N^1;\,N^0>n_1;\,{\mathcal{A}_{n_1}(\epsilon)}] \nonumber \\
& & +\:P[N^0< N^1;\,N^0>n_1;\,{\mathcal{A}_{n_1}^c(\epsilon)}].
\end{eqnarray}}
Since the universal code satisfies the stronger version of pointwise universality, for a given $\epsilon >0$, we can take $M_1$ such that $P_1(\mathcal{A}_{n_1}^c(\epsilon))=0$ for all $n_1 \geq M_1$. In the following we take $n_1 \geq M_1$.

Next consider the second term in (\ref{eq:ch6:pmd_proof_exp1}). From Proposition \ref{prop:ch6:as_st}, $P_1[N^1 < \infty]=1$ and hence,
\begin{eqnarray}
\label{eq:ch6:pmd_proof_exp2}
P_1[N^0<N^1;N^0 >n_1\,;{\mathcal{A}_{n_1}(\epsilon) }]
&\leq & P_1[N^0<\infty;N^0 >n_1\,;\mathcal{A}_{n_1}(\epsilon)].
\end{eqnarray}
\noindent Under $\mathcal{A}_{n_1}(\epsilon)$, for $n \geq n_1$, $\widehat{W}_n$ satisfies \begin{equation}
\label{eq:ch6:expn_proof_bound_llr_1}
\left(-L_n(X_1^n)-\log P_1(X_1^n)\right)+\left(\log P_1(X_1^n)-\log P_0(X_1^n)-n\frac{\lambda}{2}\right)\geq \log \frac{P_1(X_1^n)}{P_0(X_1^n)}-n\frac{\lambda}{2}-n\epsilon.
\end{equation}
R.H.S.\ is a random walk with positive drift, $D(P_1||P_0)-(\lambda/2+\epsilon)$ (since $D(P_1||P_0)>\lambda$ and $\epsilon$ is chosen $<\lambda/2$). Let $N^1_0$ be the stopping time of this random walk to cross $-|\log \beta|$. Then $P_1(n_1<{N^0} < \infty;\,\mathcal{A}_{n_1}(\epsilon))\leq P_1({N^1_0} < \infty)$. Now, from \cite[p.~79]{POO_BOOK_1},
\begin{eqnarray}
\label{eq:ch6:pmd_asymp_rw_bound}
P_1[{N^1_0}<\infty]\leq e^{s'|\log \beta|},
\end{eqnarray}
where $s'$ is the solution of $E_1[e^{s'\,\left(\log \frac{P_1(X_1)}{P_0(X_1)}-\frac{\lambda}{2}-\epsilon \right)}]=1$ and $s'<0$. 

Finally consider $P_1[N^0<N^1;N^0 \leq n_1] \leq P_1[N^0 \leq n_1]$. Since we have finite alphabet, $L_n(X_1^n) \leq M_2$ for $n=1,\ldots,n_1$ for some $M_2 < \infty$ and,
{\allowdisplaybreaks
\begin{eqnarray}
P_1[N^0 \leq n_1] &\leq & \sum_{n=1}^{n_1}P_1 \left[-L_n(X_1^n) -\log P_0(X_1^n)-n\frac{\lambda}{2} \leq -|\log \beta| \right] \nonumber \\
&\leq & \sum_{n=1}^{n_1}P_1 \left[-M_2 -\log P_0(X_1^n)-n\frac{\lambda}{2} \leq -|\log \beta| \right] \nonumber \\
\label{eq:ch6:pmd_proof_expn}
&=& \sum_{n=1}^{n_1} P_1 \left[\log P_0(X_1^n) \geq |\log \beta|-M_2-n\frac{\lambda}{2} \right]=0,
\end{eqnarray}}
for all $\beta < \beta_2$, for some $\beta_2 >0$.

Therefore as $\beta \to 0$, using (\ref{eq:ch6:pmd_proof_exp1}), (\ref{eq:ch6:pmd_proof_exp2}), (\ref{eq:ch6:pmd_asymp_rw_bound}) and (\ref{eq:ch6:pmd_proof_expn}),
\begin{equation*}
\pushQED{\qed}
P_{MD}\leq \beta^{s}= \mathcal{O}(\beta^{s}),\quad s=-s'>0.\qedhere \popQED
\end{equation*}
\vspace*{-1 cm}
\section{Proof of Theorem \ref{thm:ch6:edd_prtys}}
\label{proof:thm:ch6:edd_prtys}
(a)\quad We have,
{\allowdisplaybreaks
\begin{eqnarray}
N &=& \min \{N^0,N^1\} \nonumber \\
&=& N^0 \,\mathbb{I}\{N^0 \leq N^1\}+ N^1\, \mathbb{I}\{N^1 > N^0\}. \nonumber
\end{eqnarray}}
From Theorem \ref{thm:ch6:pfa_pmd_prtys}, $P_{FA} \to 0$ as $\alpha \to 0$, under $H_0$, and hence,
\begin{equation}
\label{eq:ch6:proof_e_0_1}
\lim_{\alpha,\beta \to 0} \frac{N}{|\log \beta|}=\lim_{\alpha,\beta \to 0} \frac{N^0\, \mathbb{I}\{N^0 \leq N^1\}}{|\log \beta|}\, \text{ a.s.}
\end{equation}

Define for, $0<r<1$, a small constant,
\begin{equation*}
\mathcal{A}_r=\{w:\sup_{n \leq N_0^*(\epsilon)} \widehat{W}_n \leq r |\log \beta|<|\log \alpha|\}.
\end{equation*}
\noindent Then, because for $n > N_0^*(\epsilon)$, $\widehat{W}_n \leq -n(\lambda/2)+n\epsilon$,
\begin{equation*}
N^0\,\mathbb{I}\{N^0 \leq N^1\}\mathbb{I}\{\mathcal{A}_r\} \leq N_0^*(\epsilon)+\frac{1+r}{\frac{\lambda}{2}-\epsilon}|\log \beta|\,\mathbb{I}\{\mathcal{A}_r\}.
\end{equation*}
\noindent Since $P_0(\mathcal{A}_r) \to 1$ as $\alpha,\beta \to 0$,
\begin{equation*}
\limsup_{\alpha,\beta \to 0} \frac{N^0}{|\log \beta|}=\limsup_{\alpha,\beta \to 0} \frac{N^0\,\mathbb{I}\{N^0 \leq N^1\}}{|\log \beta|} \leq \limsup_{\alpha,\beta \to 0} \frac{N_0^*(\epsilon)}{|\log \beta|}+\frac{1+r}{\frac{\lambda}{2}-\epsilon} \to \frac{1+r}{\frac{\lambda}{2}-\epsilon} \text{ a.s.}
\end{equation*}
Taking $r \to 0$ and $\epsilon \to 0$ we get
\begin{equation}
\label{eq:ch5:proof_e_0_2}
\limsup_{\alpha,\beta \to 0} \frac{N}{|\log \beta|}=\limsup_{\alpha,\beta \to 0} \frac{N^0\,\mathbb{I}\{N^0 \leq N^1\}}{|\log \beta|} \leq \frac{2}{\lambda} \text{ a.s.}
\end{equation}

Next define
\begin{equation*}
\mathcal{B}_r=\{w:\inf_{n \leq N_0^*(\epsilon)} \widehat{W}_n \geq -r |\log \beta|<|\log \alpha|\},
\end{equation*}
for $r$ a small positive constant $<1$. Then $P(\mathcal{B}_r) \to 1$ as $\beta \to 0$ and hence 
\begin{equation*}
\mathbb{I}\{\mathcal{B}_r\} \frac{(1-r)|\log \beta|}{\frac{\lambda}{2}+\epsilon} \leq N^0
\end{equation*}
implies
\begin{equation}
\label{eq:ch5:proof_e_0_3}
\frac{(1-r)}{\frac{\lambda}{2}+\epsilon} \leq \liminf_{\alpha,\beta \to 0} \frac{N^0}{|\log \beta|} \text{ a.s.}
\end{equation}
Taking $r \to 0$ and $\epsilon \to 0$ from (\ref{eq:ch6:proof_e_0_1}), (\ref{eq:ch5:proof_e_0_2}) and (\ref{eq:ch5:proof_e_0_3}) we get
\begin{equation*}
\lim_{\alpha, \beta \to 0} \frac{N}{|\log \beta|}=\lim_{\alpha, \beta \to 0} \frac{N^0}{|\log \beta|}=\frac{2}{\lambda} \text{ a.s.}
\end{equation*}

Observe that
\begin{equation*}
N^0 \leq N_0^*(\epsilon)+\frac{|\widehat{W}_{N_0^*(\epsilon)}|+|\log \beta|}{\frac{\lambda}{2}-\epsilon}.
\end{equation*}
\noindent Then by $C_r$-inequality, for $p \geq 1$,
\begin{equation}
\label{eq:ch5:proof_e_0_4}
E_0[(N^0)^p]\leq C_p \left[E_0[(N_0^*(\epsilon))^p]+\frac{1}{(\frac{\lambda}{2}-\epsilon)^p}(E_0[{|\widehat{W}_{N_0^*(\epsilon)}|}^p]+{|\log \beta|}^p) \right],
\end{equation}
\noindent where $C_p >0$ depends only on $p$. Also,
{\allowdisplaybreaks
\begin{eqnarray*}
E_0[|\widehat{W}_{N_0^*(\epsilon)}|^p] &=& E_0 \left[{\Big |-L_{N_0^*(\epsilon)}(X_1^{N_0^*(\epsilon)})-\log P_0(X_1^{N_0^*(\epsilon)})-N_0^*(\epsilon) \frac{\lambda}{2} \Big|}^p \right]\\
&\leq & C_p \left( E_0[({L_{N_0^*(\epsilon)}(X_1^{N_0^*(\epsilon)})})^p]+E_0[{|\log P_0(X_1^{N_0^*(\epsilon)})|}^p]+\frac{\lambda^p}{2^p} E_0[{(N_0^*(\epsilon))}^p]\right).
\end{eqnarray*}}
Furthermore, $E_0[({L_{N_0^*(\epsilon)}(X_1^{N_0^*(\epsilon)})})^p] <\infty$ if $E_0[{|\log P_0(X_1^{N_0^*(\epsilon)})|}^p] < \infty$. Since $N_0^*(\epsilon)$ is not a stopping time, for $E_0[{|\log P_0(X_1^{N_0^*(\epsilon)})|}^p] < \infty$, we need $E_0[{|\log P_0(X_1)|}^{p+1}] < \infty$ and $E_0[({N_0^*(\epsilon)})^p]\allowbreak< \infty$ (see, e.g., \cite[p.~36]{GUT_BOOK_2009}).

Thus from (\ref{eq:ch5:proof_e_0_4}), for a fixed $\epsilon$,
\begin{equation*}
\frac{E_0[({N^0})^p]}{{|\log \beta|}^p} \leq C_p \left[ \frac{E_0[{(N_0^*(\epsilon))}^p]}{{|\log \beta|}^p}+ \frac{1}{(\frac{\lambda}{2}-\epsilon)^p} \frac{E_0[{|\widehat{W}_{N_0^*(\epsilon)}|}^p]}{{|\log \beta|}^p}+1\right],
\end{equation*}
and hence $\{\frac{(N^0)^p}{|\log \beta|^p},0<\beta<1\}$ is uniformly integrable. Therefore, as $\beta \to 0$, (fix $\epsilon >0$ and then take $\epsilon \downarrow 0$)
\begin{equation*}
\frac{E_0[(N^0)^q]}{|\log \beta|^q} \to {\left(\frac{2}{\lambda} \right)}^q
\end{equation*}
and
\begin{equation*}
\frac{E_0[(N)^q]}{|\log \beta|^q} \to {\left(\frac{2}{\lambda} \right)}^q,
\end{equation*}
for all $0 <q \leq p$.\\
(b)\quad The proof for (b) follows as in (a) with the following  modifications. Interchange $N^0$ with $N^1$ in (\ref{eq:ch6:proof_e_0_1}). Use $N_1^*(\epsilon)$ instead of $N_0^*(\epsilon)$. For the convergence of moments we need $E_1[{(\log P_i (X_1))}^{p+1}] < \infty$ for $i=0,1$. 

To get the CLT result, we write
\begin{equation*}
\widehat{W}_n=[-L_n(X_1^n)-\log P_1(X_1^n)]+\sum_{k=1}^n \left( \log \frac{P_1(X_k)}{P_0(X_k)}-\frac{\lambda}{2} \right).
\end{equation*}
Under our assumptions, $\widehat{W}_n$ becomes the perturbed random walk in \cite[Chapter 6]{GUT_BOOK_2009} and from \cite[Chapter 6, Theorem 2.3]{GUT_BOOK_2009}, if $\rho^2< \infty$ and $(L_n(X_1^n)+\log P_1(X_1^n))/\sqrt{n} \to 0$ a.s., we get that, as $\alpha, \beta \to 0$,
\begin{equation*}
\frac{N^1-{\delta}^{-1}|\log \alpha|}{\sqrt{\rho^2\delta^{-3}|\log \alpha|}} \to \mathcal{N}(0,1) \text{ in distribution.} 
\end{equation*}
Now, since $|N-N^1| \to 0$ a.s., as $\alpha+\beta \to 0$, we obtain the result for $N$.\hfill \QED

\section{Proof of Theorem \ref{thm:ch3:edd}}
\label{proof:thm:ch3:edd}
We will prove the theorem under $H_0$. The proof under $H_1$ will follow in the same way.

Consider a random walk formed at the FC when the drift is $\Delta(\mathcal{A}^0)$ ($-$ve under our assumptions). We take this random walk independent of transmissions from the local nodes. Let $\nu(a)$ be the stopping time when this random walk starting at zero first crosses $a$. Then,
\[N_d \leq N^0_d \leq \tau^*(-|\log c|)+\nu(-|\log c|-F_{\tau^*(-|\log c|)+1}).\]
Therefore,
\begin{equation}
\label{eq:ch3:proof_e_0_0th_term}
\frac{N_d}{|\log c|} \leq \frac{\tau^*(-|\log c|)}{|\log c|}+\frac{\nu(-|\log c|-F_{\tau^*(-|\log c|)+1})}{|\log c|}.
\end{equation}

Now we consider the first term in \eqref{eq:ch3:proof_e_0_0th_term}. We have at local node $l$, 
\begin{equation}
\label{eq:proof_pe_x}
\frac{\tau^*_l(-r_l|\log c|)}{|\log c|} \leq \frac{N_{0,l}^*(\epsilon)}{|\log c|}+\frac{\tau_l(-(r_l|\log c|-\widehat{W}_{N_{0,l}^*(\epsilon)+1,l}))}{|\log c|}.
\end{equation}
For a fixed $\epsilon >0$, $N_{0,l}^*(\epsilon)$ is a fixed random variable with a finite mean. Therefore, the first term in R.H.S.\ approaches zero a.s.\ and in $L_1$ when $c \to 0$. Now,
\begin{eqnarray}
\frac{\tau_l(-(r_l|\log c|-\widehat{W}_{N_{0,l}^*(\epsilon)+1,l}))}{|\log c|}=\frac{\tau_l(-(r_l|\log c|-\widehat{W}_{N_{0,l}^*(\epsilon)+1,l}))}{|\log c|-\widehat{W}_{N_{0,l}^*(\epsilon)+1,l}}.\;\frac{|\log c|-\widehat{W}_{N_{0,l}^*(\epsilon)+1,l}}{|\log c|}.
\end{eqnarray}

Since $E_0\left[ \left(\log \frac{f_{1,l}(X_{1,l})}{f_{0,l}(X_{1,l})} \right)^2 \right]< \infty$ and $E[N_{0,l}^*(\epsilon)]<\infty$ for a given $\epsilon >0$, $E[|\widehat{W}_{N_{0,l}^*(\epsilon),l}|]<\infty$. Also from \cite[Remark 4.4, p.~90]{GUT_BOOK_2009} as $c \to 0$, $\tau_l(-r_l|\log c|) \to \infty$ a.s. and $\displaystyle \lim_{c \to 0}\frac{\tau_l(-r_l|\log c|)}{|\log c|}=\frac{1}{D_{tot}^0}$ a.s. for any $\epsilon>0$. Therefore,
\begin{eqnarray}
\label{eq:last-ext-convg1}
\lim_{c \to 0} \frac{\tau_l(-(r_l|\log c|-\widehat{W}_{N_{0,l}^*(\epsilon)+1,l}))}{|\log c|-\widehat{W}_{N_{0,l}^*(\epsilon)+1,l}}=\frac{1}{D_{tot}^0}\text{ a.s.}
\end{eqnarray}
We also have
\begin{eqnarray}
\label{eq:last-ext-convg2}
\frac{|\log c|-\widehat{W}_{N_{0,l}^*(\epsilon)+1,l}}{|\log c|} \to 1 \text{ a.s.}
\end{eqnarray}
From \eqref{eq:last-ext-convg1} and \eqref{eq:last-ext-convg2}, 
\begin{equation}
\label{eq:edd_dec_proof_obe_step}
\overline{\lim_{c\to 0}}\frac{\tau^*_l(-r_l|\log c|)}{|\log c|}\leq \overline{\lim_{c\to 0}}\frac{\tau_l(-(r_l|\log c|-\widehat{W}_{N_{0,l}^*(\epsilon)+1,l}))}{|\log c|} =\frac{1}{D_{tot}^0}\text{ a.s.}
\end{equation}
Also since $\left\{\frac{\tau_l(-(r_l|\log c|-\widehat{W}_{N_{0,l}^*(\epsilon)+1,l}))}{|\log c|}, c >0 \right\}$ is uniformly integrable (\cite[proof of Theorem 1 (i) $\Rightarrow$ (ii) p.~871]{Janson_AAP}), the means also converge. Thus from \eqref{eq:proof_pe_x},
\begin{equation*}
\overline{\lim_{c \to 0}}\frac{E_0[\tau^*_l(-r_l|\log c|)]}{|\log c|} \leq \frac{1}{D_{tot}^0}.
\end{equation*}
Therefore $\overline{\lim}_{c \to 0}\frac{\tau^*(-|\log c|)}{|\log c|} \leq \frac{1}{D_{tot}^0} \text{ a.s.}$ and  
\begin{equation}
\label{eq:first_term_edd_proof}
\overline{\lim_{c \to 0}}\frac{E_0[\tau^*(-|\log c|)]}{|\log c|} \leq \frac{1}{D_{tot}^0}. 
\end{equation}

The second term in R.H.S.\ of \eqref{eq:ch3:proof_e_0_0th_term},
\begin{equation}
\label{eq:ch3:proof_e_0_2nd_term}
\frac{\nu(-|\log c|-F_{\tau^*(-|\log c|)+1})}{|\log c|} \leq \frac{\nu(-|\log c|)}{|\log c|}+\frac{\nu(-F_{\tau^*(-|\log c|)+1})}{|\log c|}.
\end{equation}
\noindent We know, from (\cite[Chapter III]{GUT_BOOK_2009}), as $c \to 0$
\begin{equation}
\label{eq:ch3:proof_e_0_3rd_term}
\frac{\nu(-|\log c|)}{|\log c|} \to- \frac{1}{\Delta(\mathcal{A}^0)} \text{ a.s.\ and in $L_1$}.
\end{equation}

Next consider $\nu(-F_{\tau^*(-|\log c|)+1})/|\log c|$. It can be shown that $F_k^*$ stochastically dominates $F_k$ and thus we can construct $\{F_k^*\}$ such that $F_k^*\geq F_k$ a.s.\ for all $k \geq 0$. Let $\widehat{F}_k^*=\sum_{i=1}^k |\xi_k^*|$, $F_0^*=0$. Then,

\begin{IEEEeqnarray*}{rCl}
\frac{\nu(-F_{\tau^*(-|\log c|)+1})}{|\log c|} &\leq & \frac{\nu(-\widehat{F}_{\tau^*(-|\log c|)+1}^*)}{|\log c|}. \nonumber
\end{IEEEeqnarray*}
Also,
\begin{IEEEeqnarray*}{rCl}
\overline{\lim_{c \to 0}} \frac{\widehat{F}_{\tau^*(-|\log c|)+1}^*}{|\log c|} = \overline{\lim_{c \to 0}} \frac{\widehat{F}_{\tau^*(-|\log c|)+1}^*}{\tau^*(-|\log c|)+1} \frac{\tau^*(-|\log c|)+1}{|\log c|} \leq \frac{E[|\xi_1^*|]}{D_{tot}^0} \text{ a.s.} \IEEEeqnarraynumspace
\end{IEEEeqnarray*}
Thus,
\begin{eqnarray}
\label{eq:as_conv_entho}
\overline{\lim_{c \to 0}}\frac{\nu(-\widehat{F}^*_{\tau^*(-|\log c|)+1})}{|\log c|}= \overline{\lim_{c \to 0}} \frac{\nu (-\widehat{F}^*_{\tau^*(-|\log c|)+1})}{\widehat{F}^*_{\tau^*(-|\log c|)+1}} \frac{\widehat{F}^*_{\tau^*(-|\log c|)+1}}{|\log c|}
\leq \frac{-1}{\Delta(\mathcal{A}^0)}\frac{E[|\xi_1^*|]}{D_{tot}^0} \text{ a.s.}
\end{eqnarray}

From \eqref{eq:ch3:proof_e_0_0th_term}, \eqref{eq:edd_dec_proof_obe_step}, \eqref{eq:ch3:proof_e_0_2nd_term}, \eqref{eq:ch3:proof_e_0_3rd_term} and \eqref{eq:as_conv_entho},
\begin{equation*}
\overline{\lim_{c\to 0}} \frac{N_d}{|\log c|} \leq \frac{1}{D_{tot}^0}-\frac{1}{\Delta(\mathcal{A}^0)}\frac{E_0[|\xi_1^*|]}{D_{tot}^0} \text{ a.s.}
\end{equation*}

Now we show $L_1$ convergence.
For $\alpha>1$, 
{\allowdisplaybreaks
\begin{IEEEeqnarray}{rCl}
\frac{E_0[{\nu(-\widehat{F}^*_{\tau^*(-|\log c|)+1})}^{\alpha}]}{|\log c|^{\alpha}}
 &=& \frac{1}{|\log c|^{\alpha}} \int\limits_{0}^{|\log c|} E_0[{\nu(-x)}^{\alpha}|\widehat{F}^*_{\tau^*(-|\log c|)+1}=x]\,dP_{\widehat{F}^*_{\tau^*(-|\log c|)+1}}(x)\nonumber \\
& & \qquad + \frac{1}{|\log c|^{\alpha}} \int\limits_{|\log c|}^{\infty} E_0[{\nu(-x)}^{\alpha}]\,dP_{\widehat{F}^*_{\tau^*(-|\log c|)+1}}(x)\nonumber \\
\label{eq:ch3:asy_pty_e_proof2}
&\leq & \frac{E_0[{\nu(-|\log c|)}^{\alpha}]}{|\log c|^{\alpha}}+\int\limits_{|\log c|}^{\infty} \frac{E_0[{\nu(-x)}^{\alpha}]}{x^{\alpha}}\, \frac{x^{\alpha}}{|\log c|^{\alpha}}\,dP_{\widehat{F}^*_{\tau^*(-|\log c|)+1}}(x).\IEEEeqnarraynumspace 
\end{IEEEeqnarray}}
When $-$ve part of the increments of random walk of $\nu(t)$ has finite $\alpha^{\text{th}}$ moment (\cite[Chapter 3, Theorem 7.1]{GUT_BOOK_2009}), $E_0[{\nu(-x)}^{\alpha}]/{x^{\alpha}} \to {(-{1}/{\Delta(\mathcal{A}^0)})}^{\alpha}$ as $x \to \infty$. Thus for any $\epsilon >0$, $\exists\, M$ such that,
\[\frac{E_0[{\nu(-x)}^{\alpha}]}{x^{\alpha}} \leq \left(\epsilon+{\left(\frac{-1}{\Delta(\mathcal{A}^0)}\right)}^{\alpha}\right), \text{ for $x > M$}.\]
Take $c_1$ such that $|\log c| >M$ for $c < c_1$. Then, for $c<c_1$,
\begin{IEEEeqnarray}{rCl}
\lefteqn{\int\limits_{|\log c|}^{\infty} \frac{E_0[{\nu(-x)}^{\alpha}]}{x^{\alpha}}\, \frac{x^{\alpha}}{|\log c|^{\alpha}}\,dP_{\widehat{F}^*_{\tau^*(-|\log c|)+1}}(x)}\nonumber \\
& \leq & \frac{\epsilon+{\left(\frac{-1}{\Delta(\mathcal{A}^0)}\right)}^{\alpha}}{|\log c|^{\alpha}} \, \int\limits_{|\log c|}^{\infty} x^{\alpha}\,dP_{\widehat{F}^*_{\tau(c)+1}}(x) 
\leq \frac{\epsilon+{\left(\frac{-1}{\Delta(\mathcal{A}^0)}\right)}^{\alpha}}{|\log c|^{\alpha}}\,E_0[(\widehat{F}^*_{\tau(c)+1})^{\alpha}].\label{eq:ch3:asy_pty_e_proof3}\IEEEeqnarraynumspace
\end{IEEEeqnarray}

Since $\lim_{c \to 0} \frac{\tau^*(-|\log c|)}{|\log c|}=\frac{1}{D_{tot}^0}$ a.s. and $\{\frac{{\tau^*(-|\log c|)}^{\alpha}}{{|\log c|}^{\alpha} }\}$ is uniformly integrable, when $E[|\xi_1^*|^{\alpha+1}]<\infty$ and $E_0 \left[ \left(\log \frac{f_{1,l}(X_{1,l})}{f_{0,l}(X_{1,l})} \right)^{\alpha+1} \right] < \infty,$ $1\leq l \leq L$, we get, (\cite[Remark 7.2, p.~42]{GUT_BOOK_2009}),
\begin{equation*}
\lim_{c \to 0} \frac{E_0[(\widehat{F}^*_{\tau^*(-|\log c|)+1})^{\alpha}]}{|\log c|^{\alpha}} = \frac{E[|\xi_1^*|^{\alpha}]}{D_{tot}^0},
\end{equation*}
and
\begin{equation}
\label{eq:ch3:proof_e_0_last_term}
\sup_{c >0} \frac{E_0[(\widehat{F}^*_{\tau^*(-|\log c|)+1})^{\alpha}]}{|\log c|^{\alpha}} < \infty.
\end{equation}
From \eqref{eq:ch3:asy_pty_e_proof2}, \eqref{eq:ch3:proof_e_0_last_term}, for some $1>\delta>0$,
\begin{IEEEeqnarray*}{rCl}
\sup_{\delta>c>0}\frac{E_0[\nu(-\widehat{F}^*_{\tau^*(-|\log c|)+1})^{\alpha}]}{|\log c|^{\alpha}} & \leq & \sup_{\delta>c> 0} \frac{E_0[\nu(-|\log c|)^{\alpha}]}{{|\log c|}^{\alpha}}\\
& & \quad +\left[\epsilon+\left( \frac{-1}{\Delta(\mathcal{A}^0)} \right)^{\alpha} \right]\sup_{\delta>c> 0} \frac{E_0[(\widehat{F}^*_{\tau^*(-|\log c|)+1})^{\alpha}]}{|\log c|^{\alpha}}\\
& < & \infty.
\end{IEEEeqnarray*}

Therefore, $\{\nu(-\widehat{F}^*_{\tau^*(-|\log c|)+1})/|\log c| \}$ is uniformly integrable and hence, from \eqref{eq:as_conv_entho}, 
\begin{equation*}
\lim_{c \to 0} \frac{E_0[\nu(-\widehat{F}^*_{\tau^*(-|\log c|)+1})]}{|\log c|} \leq -\frac{1}{\Delta(\mathcal{A}^0)}.\frac{E[|\xi_1|^*]}{D_{tot}^0}.
\end{equation*}
This, with (\ref{eq:ch3:proof_e_0_0th_term}), \eqref{eq:first_term_edd_proof}, (\ref{eq:ch3:proof_e_0_2nd_term}) and \eqref{eq:ch3:proof_e_0_3rd_term}, implies that (since $\epsilon$ can be taken arbitrarily small),
{\allowdisplaybreaks
\begin{eqnarray*}
\overline{\lim_{c\to 0}} \frac{E_0[N_d]}{|\log c|} \leq \frac{1}{D_{tot}^0}+M_0,
\end{eqnarray*}}
where $ M_0=-\frac{1}{\Delta(\mathcal{A}^0)} \left[1+\frac{E_0[|\xi_1^*|]}{D_{tot}^0} \right]$. 

Similarly we can prove $ \overline{\lim}_{c\to 0} \frac{E_1[N_d]}{|\log c|}\leq\frac{1}{D_{tot}^1}+M_1$, where $M_1=\frac{1}{\Delta(\mathcal{A}^1)} \left[1+\frac{E_1[|\xi_1^*|]}{D_{tot}^1} \right].$ \hfill \QED
\section{Proof of Theorem \ref{thm:ch3:pe-bayes}}
\label{proof:thm_pe}
We prove the result for $P_{FA}$. For $P_{MD}$ it can be proved in the same way.

Probability of False Alarm can be written as,
\begin{IEEEeqnarray}{rCl}
\label{eq:ch3:proof_pfa_starting}
P_0(\text{Reject $H_0$})=P_0[\text{FA before $\tau^*(-|\log c|)$}]+P_0[\text{FA after $\tau^*(-|\log c|)$}].\IEEEeqnarraynumspace
\end{IEEEeqnarray}

Consider the first term in the R.H.S.\ of \eqref{eq:ch3:proof_pfa_starting}. It can be shown that $F_k^*$ stochastically dominates $F_k$ under $H_0$. Thus we can construct $\{F_k^*\}$ such that $F_k^* \geq F_k$ a.s.\ for all $k \geq 0$ and hence
{\allowdisplaybreaks \begin{IEEEeqnarray}{rCl}
{P_0[\text{FA before $\tau^*(-|\log c|)$}]}
&\leq &  P_0[ \sup_{0\leq k \leq \tau^*(-|\log c|)} F_k^* \geq |\log c| ]\nonumber \\
&\leq & P_0[\sum_{k=0}^{\tau^*(-|\log c|)}|\xi_k^*| \geq |\log c|].
\end{IEEEeqnarray}}
\indent Using Lemma \ref{lem:exp_finit_last_ext}, $k_1=E_0[e^{\eta \tau^*(-|\log c|)}]< \infty$. By Markov inequality,
\begin{equation}
\label{lemma:ch3:e_f_tau_tau-bound-1}
P[\tau^*(-|\log c|)>t]\leq k_1 \exp (-\eta t).
\end{equation}
Let $\widehat{F}^*_n=\sum_{k=1}^n |\xi_k^*|$. Then, with \eqref{lemma:ch3:e_f_tau_tau-bound-1}, the expected value of $|\xi_k^*|$ being positive and with exponential tail assumption of $G_0(t)$, from \cite[Theorem 1, Remark 1]{Borovkov_1995},
\begin{equation}
\label{them_pe_borokov_thm}
P_0[\widehat{F}^*_{\tau^*(-|\log c|)}>|\log c|]\leq k_2 \exp (-s_0(\eta) |\log c|),
\end{equation}
for any $0 < \eta < R_0$, where $k_2$ is a constant and $s_0(\eta)$ is defined in \eqref{eq:thm_pe_s-param}. Therefore, 
\begin{eqnarray}
\label{them_pe_borokov_thm_use}
\frac{P_0[\text{FA before $\tau^*(-|\log c|)$}]}{c} \leq k_2 \frac{c^{s_0(\eta)}}{c} \to 0,
\end{eqnarray}
if $s_0(\eta)>1$ for some $\eta$.

Now we consider the second term in \eqref{eq:ch3:proof_pfa_starting},
\begin{IEEEeqnarray*}{rCl}
P_0[\text{FA after $\tau^*(-|\log c|)$}]= P_0[\text{FA after $\tau^*(-|\log c|)$};\mathcal{A}^0]+P_0[\text{FA after $\tau^*(-|\log c|)$};(\mathcal{A}^0)^c]\IEEEeqnarraynumspace
\end{IEEEeqnarray*}
Since events $\{\text{FA after $\tau^*(-|\log c|)$}\}$ and $(\mathcal{A}^0)^c$ are mutually exclusive, the second term in the above expression is zero. Now consider $P_0\left[\text{FA after $\tau^*(-|\log c|)$};\mathcal{A}^0\right]$. For $0<r<1$,
{\allowdisplaybreaks
\begin{IEEEeqnarray}{rCl}
\lefteqn{P_0\left[\text{FA after $\tau^*(-|\log c|)$};\mathcal{A}^0\right]}\nonumber \\
 & \leq &  P_0 \Big[\text{Random walk with drift $\Delta(\mathcal{A}^0)$ and initial value $F_{\tau^*(-|\log c|)}$ crosses $|\log c|$}\Big] \nonumber\\
 & \leq & P_0 \Big[\text{Random walk with drift $\Delta(\mathcal{A}^0)$ and $F_{\tau^*(-|\log c|)+1} \leq r |\log c|$ crosses $|\log c|$}\Big] \nonumber \\
 & & +\: P_0 \Big[\text{Random walk with drift $\Delta(\mathcal{A}^0)$ and $F_{\tau^*(-|\log c|)+1} > r |\log c|$ crosses $|\log c|$}\Big] \nonumber\\
&\leq & P_0 \Big[\text{Random walk with drift $\Delta(\mathcal{A}^0)$ and $F_{\tau^*(-|\log c|)+1} \leq r |\log c|$ crosses $|\log c|$}\Big] \nonumber \\
\label{eq:ch3:pfa_proof_1}
& & +\: P_0\left[F_{\tau^*(-|\log c|)+1} > r |\log c|\right].
\end{IEEEeqnarray} }
\noindent Considering the first term in the above expression,
{\allowdisplaybreaks
\begin{IEEEeqnarray}{rCl}
\lefteqn{P_0 \big[\text{Random walk with drift $\Delta(\mathcal{A}^0)$ and $F_{\tau^*(-|\log c|)+1} \leq r |\log c|$ crosses $|\log c|$}\big] /c}\nonumber \\
& \leq & P_0 \big[\text{Random walk with drift $\Delta(\mathcal{A}^0)$ and $F_{\tau^*(-|\log c|)+1} = r |\log c|$ crosses $|\log c|$}\big] /c \nonumber \\
 \label{eq:ch3:proof_pfa_2nd_term_final}
 & \stackrel{(A)}\leq & \frac{\exp(-(1-r)|\log c|s')}{c}= \frac{c^{(1-r)s'}}{c} \to 0,
\end{IEEEeqnarray}}
iff  $(1-r) s' > 1$. Here $(A)$ follows from \cite[p.~78-79]{POO_BOOK_1}
\footnote{For a random walk $W_n=\sum_{i=1}^n X_i$, with stopping times $T_{a}=\inf\{n\geq 1:W_n \leq a \}$, $T_{b}=\inf\{n\geq 1:W_n \geq b\}$ and $T_{a,b}=\min(T_a,T_b)$, $a <0 <b$, let $s'$ be the non-zero solution to $M(s')=1$, where $M$ denotes the M.G.F.\ of $X_i$. Then, $s'<0$ if $E[X_i] >0$, and $s' >0$ if $E[X_i] <0 $ and $E[\exp(s'W_{T_{a,b}})]=1$ (\cite[p.~78-79]{POO_BOOK_1}). Then it can be shown that $P(W_{T_{a}}) \leq \exp(-s'a)$ when $E[X_i] >0$ and $P(W_{T_{b}}) \leq \exp(-s'b)$ when $E[X_i] <0$.}
 where $s'$ is positive and it is the solution of $E_0\Big[e^{s'\,\log \frac{g_{\mu_1}(Y_k)}{g_{-\mu_0}(Y_k)}}|\mathcal{A}^0\Big]=1.$ 
 
We choose $s'>1$ and $0<r<1$ to satisfy $(1-r)s'> 1$. 

Consider the second term in (\ref{eq:ch3:pfa_proof_1}). Using the stochastical dominance of $\{F_k\}$ by $\{\widehat{F}^*_k\}$,
\begin{eqnarray}
\label{eq:ch3:pfa_proof_2}
P_0\left[F_{\tau^*(-|\log c|)+1} > r |\log c|\right] & \leq & P_0\left[\widehat{F}^*_{\tau^*(-|\log c|)+1} > r |\log c|\right] \nonumber
\end{eqnarray}
We have, $P[\tau^*(-|\log c|)+1>t]=P[\tau^*(-|\log c|)>t-1]\leq k_1' \exp (-\eta t)$, where $k_1'=e^{\eta}E_0[e^{\eta \tau^*(-|\log c|)}]$. Therefore, following \eqref{them_pe_borokov_thm},
\begin{eqnarray}
\label{eq:ch3:proof_pfa_3rd_term_final}
\frac{P_0\left[F_{\tau^*(-|\log c|)+1} > r |\log c|\right]}{c} &\leq & k_2' \frac{c^{rs_0(\eta)}}{c} \to 0,
\end{eqnarray}
if $rs_0(\eta)> 1$ and $k_2'$ is a constant. We can choose $s_0(\eta) >1$ as in \eqref{them_pe_borokov_thm_use}. Then $\displaystyle \frac{1}{s_0(\eta)} < r \leq 1-\frac{1}{s'}$. 

This proves Theorem \ref{thm:ch3:pe-bayes}(a). Theorem \ref{thm:ch3:pe-bayes}(b) can be proved in a similar way. \hfill \QED
\begin{lemma}
\label{lem:exp_finit_last_ext}
$E_0\left[e^{\eta \, \tau^*(-|\log c|)}\right]<\infty$ for $0< \eta<R_0$.	
\end{lemma}
\begin{IEEEproof}
We have,
\begin{equation}
\tau_l^*(-r_l|\log c|) \leq N_{0,l}^*(\epsilon)+\tau_l(-(r_l|\log c|-\widehat{W}_{N_{0,l}^*(\epsilon)+1,l})).
\end{equation}
Therefore, for $\eta>0$,
\begin{eqnarray}
E_0\left[e^{\eta \, \tau^*_l(-r_l|\log c|)}\right] &\leq & E_0\left[e^{\eta N_{0,l}^*(\epsilon)+\eta \tau_l(-r_l|\log c|+\widehat{W}_{N_{0,l}^*(\epsilon)+1,l})}\right] \nonumber \\
\label{eq:fini_exp_mome_lastexi}
& \leq & E_0\left[e^{p\eta N_{0,l}^*(\epsilon)}\right]^{1/p}.\; E_0\left[e^{q\eta \tau_l(-r_l|\log c|+\widehat{W}_{N_{0,l}^*(\epsilon)+1,l})}\right]^{1/q},
\end{eqnarray}
where $p>1$ and $1/p+1/q=1$. We have under strong version of pointwise universality \eqref{eq:stron_ver_po_uni}, $P[N_{0,l}^*(\epsilon)=m]=0$ for all $m\geq n_1$ for some $n_1$. Therefore, $E[e^{p\eta N_{0,l}^*(\epsilon)}]<\infty$ for all $\eta >0$ and $p>0$. From \cite[Theorem 1.3]{Iksanov_ECP} $E_0[e^{q\eta \tau_l(-r_l|\log c|+\widehat{W}_{N_{0,l}^*(\epsilon)+1,l})}]$, for $0< \eta<\eta q< R^l_0$ and $ R^l_0=-\log \inf_{t \geq 0} E_0\Big[e^{-t \log \frac{f_{1,l(X_{1,l})}}{f_{0,l}(X_{1,l})}}\Big]$ where $q>1$ can be taken as close to one as needed. Thus, $E[e^{\eta \, \tau^*_l(-r_l|\log c|)}]<\infty$ for $0<\eta<R^l_0$. Combining this fact with $\tau^*(-|\log c|)< \sum_{l=1}^L \tau_l^*(-r_l|\log c|)$ (see (\ref{eq:ch3:proof_notn_tau})) yields $E_0[e^{\eta \tau^*(-|\log c|)}] < E_0[e^{\sum_{l=1}^L \eta \tau_{l}^*(-r_l|\log c|)}]<\infty$, for $0< \eta < R_0=\min_l R^l_0$ because $\tau_l^*(-r_l|\log c|),l=1,\ldots,L$ are independent.
\end{IEEEproof}
\bibliographystyle{IEEEtranS}
\bibliography{IEEEabrv,mybib_ITA}
\end{document}